\documentclass[cmp]{svjourmod}

 \pdfoutput=1

\usepackage[pdftex]{graphicx}
\usepackage{amssymb}
\usepackage{bm}
\usepackage{psfrag}
\usepackage{amsmath}
\usepackage{amsfonts}
\usepackage{mathrsfs}
\usepackage{cite}

\usepackage{xcolor}

\newcommand{\p}{\partial}

\newcommand{\dd}{{\rm d}}
\begin{document}

\title{Surface gravity of compact non-degenerate horizons under the dominant energy condition}
\titlerunning{Surface gravity of compact non-degenerate horizons}

\author{Sebastian Gurriaran\thanks{{\'E}cole Normale Sup\'erieure (ENS) Paris-Saclay \email{sebastian.gurriaran@ens-paris-saclay.fr}} and Ettore Minguzzi\thanks{Dipartimento di Matematica e Informatica ``U. Dini'', Universit\`a degli
Studi di Firenze,  Via S. Marta 3,  I-50139 Firenze, Italy. \email{ettore.minguzzi@unifi.it}}}

\institute{}

%
%
%

\date{}

\maketitle

\begin{abstract}
\noindent
We prove that under the dominant energy condition any non-degene\-ra\-te smooth compact  totally geodesic horizon admits a smooth tangent vector field of constant non-zero surface gravity. This result generalizes previous work by Isenberg and Moncrief, and by Bustamante and Reiris to the non-vacuum case, the vacuum case being given a largely independent proof. Moreover, we prove that any such achronal non-degenerate horizon is actually a Cauchy horizon bounded on one side by a chronology violating region.
\end{abstract}


\section{Introduction}

In this work we are going to study the surface gravity over compact null hypersurfaces.
This study is motivated by the strong cosmic censorship conjecture. Roughly speaking, this conjecture states that under reasonable physical assumptions, spacetimes have to be globally hyperbolic. This means that no complete spacelike hypersurface $S$ should have a non-empty  Cauchy horizon $H^+(S)$.

Naturally, one could hope to start proving the compact case of the conjecture, which means to show that under reasonable conditions compact Cauchy horizons cannot form.
Notice that the cosmic censorship conjecture is framed in the context of general, non-necessarily vacuum spacetimes. The rigidity case of this problem asks to prove that for vacuum spacetimes (the vacuum condition can be regarded as a non-generic  condition) compact Cauchy horizons can form but just in special circumstances.\footnote{As another example of a rigidity case, consider Hawking-Penrose's singularity theorem. Its rigidity case is the Lorentzian splitting theorem, which states that under the assumptions of H.-P.'s theorem, provided we drop the genericity condition,   there can be a complete timelike line, but only under a very special circumstance: the spacetime splits as a product.}

The precise formulation of this rigidity case  was given by Isenberg and Moncrief \cite{moncrief83}.
They conjectured that every smooth compact null hypersurface (hence generated by lightlike geodesics) in a (electro-)vacuum spacetime is actually Killing, namely, there exists a Killing vector field tangent to the hypersurface \cite{moncrief83}. They were able to  prove the result under the additional assumptions \cite{isenberg85}, (a)  the metric and horizon are analytic,\footnote{Now we know that the analiticity of the horizon follows from that of the metric \cite{minguzzi14d}.} and (b) the horizon generators are closed; but conjectured that (a) and (b) could be dropped.

In a subsequent remarkable paper \cite{moncrief08} they introduced the ribbon argument. With it they were able to prove, under the assumption of the existence of a certain atlas  \cite[Eq.\ (4.1)]{moncrief08} (which implies the existence of an integrable distribution), under analyticity,  and still in the vacuum  case, that if the horizon admits a future incomplete generator, then all generators are actually future incomplete. Unfortunately, the existence of a foliation (integrable distribution) transverse to the generators is quite a non-trivial matter, the more so in the analytic case, a fact which excluded some types of geodesic dynamics from their analysis.
This problem was  addressed by Bustamante and Reiris in \cite{reiris21}.
In the smooth category they were able to replace the coordinate slices (foliations) of \cite{moncrief08,moncrief20} with  a  `horizontal' not-necessarily-integrable distribution by introducing the `horizontal exponential map'.

Let $n$ be a future-directed lightlike vector field tangent to the horizon $H$. Here the surface gravity is the ($n$-dependent) function $\kappa: H \to \mathbb{R}$
\[
\nabla_n n=\kappa n.
\]
Isenberg and Moncrief were able to show \cite{moncrief83}, under analyticity, that the  surface gravity $\kappa$ has to be constant and that the cases $\kappa\ne 0$ (non-degenerate) and $\kappa =0$ (degenerate) should be studied separately.

The former would correspond to generators incomplete in one direction and the latter to complete generators. Actually, they showed that any compact analytic null hypersurface with constant non-zero surface gravity and closed generators is  a Cauchy horizon \cite{moncrief83}.

The analytic limitation prevented the proof of their original conjecture. In \cite{friedrich99} Friedrich, R\'acz and Wald proved the existence of the Killing vector on one side of the  horizon. They worked in the smooth category and under a vacuum assumption but still imposed the closure of the generators. In a subsequent work R\'acz showed that, under the same assumptions, the constancy of surface gravity could be proved already for the non-vacuum case provided the dominant energy condition was imposed \cite[Prop.\ 4.1]{racz00}.

In a series of recent papers Petersen \cite{petersen18b,petersen19} and Petersen and R\'acz \cite{petersen18} have  shown that in the vacuum smooth  case and for smooth horizons it is possible to prove the existence of a Killing vector field tangent to the horizon provided a {\em smooth} choice for $n$ exists such that $\kappa=-1$. Finally, Bustamante and Reiris \cite{reiris21} proved, by using a modified ribbon argument and the horizontal exponential map mentioned above, that in the vacuum smooth case it is indeed possible to choose $\kappa=-1$. This solved the Isenberg--Moncrief conjecture in the non-degenerate case.

Of course, these results have  very important consequences, as they fix considerably the possible topologies and symmetries of a spacetime admitting a compact Cauchy horizon \cite{isenberg92,moncrief20,bustamante21}.

The purpose of this work is to study the constancy of surface gravity over compact null hypersurfaces without the vacuum assumption.
We shall be able to generalize most of the previous results to the non-vacuum case under a dominant energy condition. This result can be obtained rather easily, provided some key observations are made.



Under the dominant energy condition we shall  prove that compact Cauchy horizons admitting an incomplete geodesic generator (i.e.\ non-degenerate) do admit a smooth field $n$ with constant surface gravity normalized to $-1$.
In a sense what we show in this paper is that these compact horizons  are rigid.
The Einsten equations in vacuum serve just   to propagate this rigidity well outside the horizon.
Already in the vacuum case our proof is alternative that  of \cite{reiris21}. Given the importance of the results in \cite{reiris21} we find it useful to present a different approach and proof, in fact most of the paper will be devoted to it.

The main difference with respect to \cite{reiris21} is that, while they work with a vector field $Z$ defined just on the ribbon and proportional to $n$, we only work with a smooth vector field $n$ which is globally defined over $H$.  Compactness arguments become particularly transparent in our approach and we do not need to introduce the horizontal exponential map or some frame bundle. The ribbon argument is used just once, and then the ribbon is not mentioned anymore, contrary to \cite{reiris21} where it is present through the whole proof.

Letting $\Lambda(p)$ be the geodesic length of the generator starting at $p$ with tangent $n(p)$, our approach not only allows us to prove that $\Lambda$ is smooth, it also allows us to compute the first and second covariant derivatives of $\Lambda$, and to show that the latter satisfies an elliptic PDE (in the vacuum case). The smoothness of $\Lambda$, which is key to show the existence of a smooth $n$ with constant $\kappa$, follows then immediately by a bootstrap argument. The non-vacuum case requires a slightly longer analysis.

Another difference is that, while they basically prove constancy of $\kappa$ and smoothness of $\Lambda$ at once, we first prove that $n$ can be chosen so that  $\kappa$ is  negative and bounded away from zero, and then use this property to show that the integral expression for  $\Lambda$ can be differentiated. In the former step we are also able to obtain a rather clear connection between the divergent behavior of the integral $\int \kappa(s)\dd s$ along one generator and the incompleteness of the generators (cf.\ the equivalence between 1 and 4 in Def.\ \ref{der}).

Although longer, our proof provides some fine details on the analysis of the problem and also leads to some interesting equations in the non-vacuum case (cf.\ Thm.\ \ref{puy}, Sec.\ \ref{nnn}, Eq.\ (\ref{cqy})).

{
\begin{remark}
Anticipating some of the notation to be introduced later on, we might more precisely outline the proof strategy as follows. If every generator of the horizon starting with velocity $n$  is future incomplete, and if the finite affine length $\Lambda(p)$ depends smoothly on the starting point $p$ of the generator, then the  vector field $K=\Lambda n$ is easily shown to have surface gravity $-1$.
Unfortunately,  the affine length has expression (Lemma \ref{cpx}) $\Lambda=\int_0^\infty \exp (\int_0^\rho \kappa)d \rho$, and it is difficult to control convergence and regularity under change of generator.

Our strategy consists in  showing that there is a sufficiently large constant $C>0$, independent of the generator, such that, if we take instead the approximating expression  $n'= (\int_0^C \exp (\int_0^\rho \kappa) d \rho ) n$ we get {\em negative} surface gravity $\kappa'$ for $n'$ (Thm. \ref{bhq}) provided just one generator is future incomplete.
It turns out that we just need $C$ as in the result of Proposition \ref{viu} for this to work. Thus we follow the steps:
\begin{enumerate}
\item If a generator is incomplete then over the same curve ($s$-parametrized, with $n=d/ds$)  we have  $\int_0^\infty \kappa(s)\dd s =-\infty$.
\item Using the ribbon argument, if the equation  $\int_0^\infty \kappa(s)\dd s =-\infty$ is true for one generator, then it is true for every generator.
\item Using this fact, and compactness, on a non-degenerate horizon we can get a global $C>0$ such that for $\rho\geq C$, and for every generator, $\int_0^\rho \kappa(s) \dd s <0$.
\item The vector field $n'=(\int_0^C \exp (\int_0^s\kappa) d s ) n$ has negative surface gravity, and hence every generator is future incomplete.
\item  $\Lambda'$, the affine length obtained using $n'$, is now finite. Its smoothness follows from the fact that $\Lambda'$ has an integral expression, together with its derivatives, which involves the exponential  $e^{\int_0^\rho \kappa'(s)\dd s}$ with $\kappa'$ upper bounded by a negative constant (by continuity and compactness). Finally, the field $\Lambda' n'$ is smooth and has surface gravity $-1$.
\end{enumerate}
\end{remark}

}

Continuing with the description of the paper, in Section \ref{nnn} we are able to prove, in any dimensions, a result previously obtained by Moncrief and Isenberg under assumptions of analyticity and closure of generators, namely that in the non-degenerate case horizons are in fact Cauchy horizons bounded by regions of chronology violation, cf.\ \cite{moncrief20} and \cite[Cor.\ 2.13]{petersen19} for the   vacuum case. We also argue that the  horizon classifications by Rendall \cite{rendall98}, and by Bustamante and Reiris \cite{bustamante21} can be extended to the non-vacuum case.

It is pleasing that the constancy of surface gravity on (non-degenerate) compact horizons can  be proved by just imposing the dominant energy condition as this result becomes completely analogous to that on the constancy of surface gravity on (non-extremal) Killing  horizons in Black Holes physics \cite{bardeen73} \cite[Sec.\ 12.5]{wald84} \cite[Thm.\ 7.1]{heusler96} \cite[Thm.\ 4.3.12]{chrusciel20}. In fact, the latter result can be obtained from the former by using the usual compactification trick as in \cite[Sect.\ 2]{friedrich99} (the other direction in not possible). As mentioned, this type of results appeared in the literature but only under the assumption of closure of generators \cite{racz00}.

Finally, it is worth mentioning that this work is much self contained and that all computations are coordinate independent. For instance, we do not need to introduce Gaussian null coordinates. Hopefully, this approach, which we find
particularly
efficient, will be appreciated by some readers.

Let us introduce some notations and conventions. The spacetime $(M,g)$ is a smooth connected time-oriented  Lorentzian manifold of dimension $n+1$. The signature of $g$ is $(-, +, \cdots, +)$. The symbol of inclusion is reflexive, $X\subset X$. The compact null hypersurface or the compact Cauchy horizon is denoted $H$,  smooth future-directed lightlike vector fields tangent to $H$ are denoted $n$ (this should not cause confusion with the manifold dimension), their integral parameter is denoted with $s$ or $\rho$, and their flow is denoted $\varphi_s$.
 The reader is referred to \cite{minguzzi18b} for all the other conventions adopted without mention in this work.

\section{Mathematical preliminaries}

We recall that the {\em null (energy) convergence condition} is: for every null vector $X\in TM$, $\textrm{Ric}(X,X) \ge 0$. By {\em ($\Lambda$-)dominant energy condition} we shall understand the following property: Let $\Lambda\in \mathbb{R}$ be a constant (the cosmological constant). The endomorphism $T: TM \to TM$ given by
\begin{equation} \label{fir}
T:\ X\mapsto \{\textrm{Ric}(X,\cdot)-[\frac{1}{2} R+\Lambda] g(X,\cdot)\}^\sharp
\end{equation}
sends the future timelike causal cone into the past causal cone (and, by continuity, the future causal cone into the past causal cone). Since the scalar product of a past causal vector with a future causal vector is non-negative, the dominant energy condition implies the null convergence condition \cite{hawking73}. Of course, under the Einstein's equations, $T$ can be interpreted as the stress-energy tensor.

A $C^2$  hypersurface is known to be null iff locally achronal and ruled by lightlike geodesics \cite{kupeli87,galloway00} \cite[Thm.\ 6.7]{minguzzi18b}. This property suggests how to define the notion of
$C^0$ null hypersurface \cite{galloway00}. Precisely a {\em past $C^0$ null hypersurface } $H$ is a locally achronal topological
embedded hypersurface such that for every $p\in H$ there is a past inextendible lightlike geodesic contained in $H$  with future endpoint $p$. Past $C^0$ null hypersurface will also be called  {\em future horizons}.
The  most notable example of future horizon is the future Cauchy horizon $H^+(S)$ of a partial Cauchy hypersurface  $S$.

Horizons can be highly irregular \cite{chrusciel98}, and this also extends to compact Cauchy horizons \cite{budzynski03}. However, under the null energy condition we have \cite[Thm.\ 18]{minguzzi14d} and \cite[Thm.\ 1.18]{larsson14}

\begin{theorem} \label{mai}
Suppose that the null convergence condition holds. Let $H$ be a compact
achronal
\footnote{Incidentally, with Prop.\ \ref{ptg} we shall prove that {\em achronal} can be dropped.}
future horizon whose past inextendible generators are past complete.  Then  $H$ is an embedded  smooth null hypersuface
which is analytic if the metric is analytic. Moreover, $H$ is totally geodesic and it is generated by inextendible lightlike geodesics.
\end{theorem}

This result was phrased in terms of Cauchy horizons in \cite{minguzzi14d,larsson14} (but Theorem \ref{mai} is immediate from the proofs).
Indeed, the important application to future Cauchy horizons was possible because of the proof that future Cauchy horizons (also non-smooth ones) have past complete generators \cite{minguzzi14,krasnikov14} (the original proof  by Hawking  \cite{hawking73} relied on a smoothness assumption).


It should be mentioned that for $C^2$ null hypersurfaces the relationship between the total geodesic property and the second fundamental form is somewhat non-trivial. The second fundamental form is defined as an endomorphism $X\mapsto b(X):= \overline{\nabla_Xn}$ on the quotient bundle $TH/n$, where the overline denotes projection under that quotient. It can be shown that it vanishes iff $\nabla_XY\in TH$ for any two vector fields tangent to the horizon. This equivalence goes back to Kupeli \cite[Thm.\ 30]{kupeli87}.

\section{Smooth null hypersurfaces}
In this paper we shall be mostly interested in the null hypersurfaces that
result from the application of Thm.\ \ref{mai}. Since we shall not use all the properties implied by that theorem, we shall assume throughout the paper and without further notice  that

\begin{itemize}
\item[($\star$)] {\em $H$ is an embedded compact connected smooth totally geodesic null hypersurface.}
\end{itemize}
By a well known result, smoothness implies that $H$ is generated by lightlike geodesics \cite{beem98,chrusciel98b,minguzzi14}. The total geodesic property implies that the expansion and shear vanish, $\theta=\sigma=0$, and then the Raychaudhuri equation implies $R(n,n)=0$ over $H$.

Since $H$ is totally geodesic  there is a 1-form $\omega: H\to T^*H$ such that
\[
\nabla_Xn=\omega(X) n
\]
where $X\in TH$.


\begin{definition}
The function $\kappa: H \to \mathbb{R}$, determined by $\nabla_nn=\kappa n$, i.e.\ $\kappa:=\omega(n)$,
is called {\em surface gravity}.
\end{definition}

Of course, the surface gravity depends on the chosen smooth field $n$. Our problem with be to determine if a smooth choice of vector field with constant surface gravity exists.
 \vspace{0.2cm}

We provide a couple of  examples of vacuum spacetimes admitting a compact Cauchy horizon, and show how to conveniently calculate their surface gravity. For more examples, with classification results for topologies and dynamical behavior of generators, see \cite{bustamante21,kroenke21}.
\begin{example}[Quotient Schwarzschild spacetime]
Let us consider, for some constant $m>0$, the
Schwarzschild spacetime  in ingoing Eddington-Finkelstein coordinates. This is the manifold $\tilde M=(0,\infty)_r\times \mathbb{R}_v
\times S^2$ endowed with the metric
\[
g = 2drdv -
(1-\frac{2m}{r})
dv^2 + r^2g_{{S}^2},
\]
and time orientation provided by the (future-directed) lightlike vector $T:=-\p_r$.
Observe that this is a coordinate patch that does not cover the whole Kruskal-Szekeres maximal extension, e.g.\  the white hole horizon is excluded.  The black hole horizon is the region $r=2m$, the interior of the black hole is the region $r<2m$, and the exterior is the region $r>2m$. The vector field $n:=\p_v$ is Killing, tangent to the horizon and timelike on the exterior of the black hole. Let $\varphi$ be the flow of $n$, and for some constant $c>0$ let us identify every  $p\in \tilde M$ with $\varphi_{nc}(p)$ for every $n\in \mathbb{Z}$. The spacetime $M$  obtained through this identification is a vacuum spacetime which admits the compact Cauchy horizon $H=\{r=2m\}$ of topology $S^2 \times S^1$. The Killing vector $n$  is tangent to $H$.

The fastest way to calculate the surface gravity is via Proposition \ref{bkq} that we shall prove later on (but the reader can already check its simple proof). The vector field $T=-\p_r$ commutes with $n$ and is such that $g(T,n)=-1=:-\frac{1}{a}$ thus the surface gravity is
\[
\kappa=\frac{a}{2} \p_T g(n,n) \vert_H=-\frac{1}{2} \p_r (\frac{2m}{r}-1) \vert_H=\frac{1}{4m}.
\]
The identification sends the exterior $r>2m$ into a region of chronology violation that lies in the causal past of $H$, while a compact partial Cauchy hypersurface $S$ with past Cauchy horizon $H^-(S)=H$ can be found in the causal future of $H$ (a general theorem establishing these features will be proved in Theorem \ref{qkw}, see also the references cited in that section). It should be noted that the generators of $H$ are past incomplete because in the covering $\tilde M$ the white hole and hence the bifurcation region of the horizon present in the Kruskal-Szekeres maximal extension, that they would intersect, has been cut out \cite{boyer69}.
\end{example}

\begin{example}(Misner spacetime)
The Misner vacuum spacetime is given by the manifold $M=\mathbb{R}_t\times S^1_x\times T^2$ endowed with the metric
\[
g=-2dtdx + tdx^2 + g_{T^2},
\]
and time orientation provided by the (future-directed) lightlike vector $T:=\p_t$.
The region of chronology violation is given by the open set $t<0$, while $H=\{t=0\}$ is a null hypersurface of topology $T^3$ which is the past Cauchy horizon of the spacelike partial Cauchy hypersurfaces $S_c=\{t=c\}$, $c>0$.
Note that $n:=\p_x$ is Killing, tangent to $H$, timelike in the region $t<0$ and spacelike in the region $t>0$.

Since $g(n,T)=-1=:-\frac{1}{a}$ we have, again by Proposition \ref{bkq}, that  the surface gravity is
\[
\kappa=\frac{a}{2} \p_T g(n,n) \vert_H=\frac{1}{2} \p_t (t) = \frac{1}{2}.
\]
Affinely parametrized geodesics on the horizon have the form $\lambda\mapsto (t(\lambda),x(\lambda), q)=(0,2\log(\frac{1}{2} (\lambda-\bar \lambda))+cnst,q)$, $q\in T^2$, and hence are past incomplete.
\end{example}

{}\vspace{0.2cm}

We noted that there are many possible choices for $n$. We observe
\begin{lemma}
For any function $f$, redefined $n'= e^f n$, we have that $\omega'=\omega+\dd f$.
\end{lemma}

\begin{proof}
\[
\omega'(X) e^fn=\nabla_X (e^f n)=e^f X(f) n+e^f \nabla_X n=[\p_X  f+\omega(X) ] e^f n
\]
$\square$ \end{proof}
It can be observed that this is the typical gauge transformation of the potential for gauge theories based on the commutative group $(\mathbb{R},+)$, cf.\ \cite{kobayashi63}.
Indeed we have a half-line bundle $L\to H$, the fiber of the point $p$ being given by the future-directed lightlike vectors tangent to $H$ at $p$.
$L$ is diffeomorphic to $H\times \mathbb{R}$, via the map $e^s n(p) \to  (p, s)$, where the vector field $p\to n(p)$ provides the global section. This shows that the bundle is trivial (as it is always the case for the contractible fiber $\mathbb{R}$).

Furthermore, on it we have an Abelian connection, indeed
\[
\tilde \omega:=\dd s+\omega
\]
is a connection on $L$ (it is invariant under vertical translations and gives 1 on $\p/\p s$). Let $\Lambda(p)\in (0,\infty]$ be the affine length of the geodesic $\gamma$ such that $\gamma(0)=p$, $\dot\gamma(0)=n$. Observe that under a  change of section (gauge transformation)
\begin{align}
s'&=s-f,  \label{ok1} \\
n'&=e^f n, \label{ok2} \\
\omega'&=\omega + \dd f, \label{ok3}\\
\Lambda'&=e^{-f}\Lambda,  \label{ok4} \\
\kappa'&= e^f(\kappa + \p_n f), \label{ok5}
\end{align}
 while
\begin{align}
\tilde \omega:=&\dd s+\omega,\\
\Omega:=&\dd \omega,\\
K:=&\Lambda n , \\
U:=& \kappa \Lambda+\p_n \Lambda,
\end{align}
are left unchanged, so are independent of the section (an important invariant built from the second derivatives of $\Lambda$  will be introduced later on, cf.\ Eq.\ (\ref{kkv})).
 Here we have extended the half-line bundle $L$ into a bundle $\bar L$ so as to include the infinite future lightlike vector, and so that $K$ is a section of $\bar L\to H$ (as it can be $\Lambda=\infty$).  By Eq.\ (\ref{ok4}) at those $p$ where $\Lambda(p)$ is finite, $K(p)$ is that unique tangent such that the future-directed geodesic generator $\gamma(0)=p$, $\dot \gamma(0)=K$, has affine length one. In our paper $n$ will be always smooth so  $K$ will be finite and smooth provided we can prove that $\Lambda$ is finite and smooth. This is the strategy of \cite{moncrief20,reiris21} which we shall also follow.

\begin{lemma} \label{cpx}
Let $x:\mathbb{R}\to H$ , $s \mapsto x(s)$,  be an integral curve of $n$, with $x(0)=p$.
The geodesic $\gamma$ starting from $x(\tau)$ with tangent $\dot \gamma=n$ has future affine length
\begin{equation} \label{bud}
\Lambda(x(\tau))=\int_\tau^\infty e^{\int_\tau^\rho \kappa(x(s))\dd s} \dd \rho.
\end{equation}
Thus $\Lambda$ is finite at one point of the integral curve iff it is finite everywhere over it, and in this case it satisfies the differential equation
\begin{equation} \label{biw}
1+\kappa \Lambda+\p_n \Lambda=0.
\end{equation}
In particular, for $\kappa$ constant, we have $\Lambda<\infty$ iff $\kappa<0$, in which case $\Lambda=
- \kappa^{-1}$.
\end{lemma}

\begin{proof}
Let  $x:\mathbb{R}\to H$ , $s \mapsto x(s)$,  be an integral curve of $n$, with $x(0)=p$, and let $f(s)$ be such that $\dot \gamma(t(s))=f(s) n( x(s))$.  Note that $f(\tau)=1$. The geodesic condition reads
\[
0=\nabla_{\dot \gamma} \dot \gamma=f[f'+f \kappa] n
\]
thus $f(s)=C \exp[- \int_0^s \kappa(x(s)) \dd s]=\exp[- \int_\tau^s \kappa(x(s)) \dd s]$, where we used the initial condition, thus $\frac{\dd t}{\dd s}=f^{-1}(s)=\exp (\int_\tau^s \kappa(x(s)) \dd s)$ and hence Eq.\ (\ref{bud}).

As a consequence,
\[
\p_\tau \Lambda(x(\tau))=-1-\int_\tau^\infty \kappa(x(\tau)) \exp^{\int_\tau^\rho \kappa(x(s))\dd s} \dd \rho=-1- \kappa(x(\tau)) \Lambda(x(\tau)).
\]
$\square$ \end{proof}

\begin{remark}
As mentioned, the main problem will be to show that $K$ is a  {\em smooth finite} section of the bundle $\bar L\to H$, hence a section of $L\to H$. The existence of such a privileged section $K$ will not necessarily imply that the bundle is flat, in fact in general we shall have  $\dd \omega\ne 0$.


\end{remark}



We need the following algebraic result of which we provide a coordinate independent proof (as mentioned, in this work we do not need to introduce Gaussian null coordinates, compare \cite{petersen18,petersen18b,moncrief83}).

\begin{lemma} \label{vie}
For vector fields $X,Y\in TH$
\begin{align*}
R(X,Y)n&=\dd \omega(X,Y) n, \\
Ric(Y,n)&=\dd \omega(n,Y) .
\end{align*}
Moreover, there is a bilinear form $\mu: H \to T^*H\otimes T^*H$ on the horizon such that for every $X,Y\in TH$, $R(n,X)Y=\mu(X,Y) n$ and $\mu(X,Y)-\mu(Y,X)=-\dd \omega(X,Y)$.
\end{lemma}

The result establishes that the 2-form $\dd \omega$, which is independent of the section $n$, has  indeed an important geometrical meaning. The first formula establishes the connection between the curvature of $\nabla$ and the curvature $\dd \omega$ of the Abelian gauge theory introduced previously.

Note that the right-hand side makes only sense for $X,Y\in TH$ as $\omega\in T^*H$ while the left-hand side could make sense also for $X,Y\in TM$.

One could also calculate the second Bianchi identity by using the first expression for the curvature, just to check whether there are further conditions to be imposed on $\omega$. It turns that this identity is equivalent to the closure of $\dd \omega$ and so it is trivially satisfied.

\begin{proof}
Indeed, remembering that $\nabla$ is well defined as a Koszul connection on $H$
\begin{align*}
\nabla_X\nabla_Yn-\nabla_Y\nabla_Xn-\nabla_{[X,Y]}n&=\nabla_X [\omega(Y) n]-\nabla_Y [\omega(X) n]- \omega([X,Y]) n\\
&=\{\nabla_X [\omega(Y)]-\nabla_Y [\omega(X)]-\omega([X,Y])\}n \\
&=(\dd \omega)(X,Y) n.
\end{align*}

Let $N$ be a lightlike vector field on $H$ such that  $g(N,n)=-1$.
Let $e_1,\cdots,$ $ e_{n-1}\in T_pH$ be such that $\textrm{Span}(e_1, \cdots, e_{n-1})=\textrm{ker}g(N,\cdot)\cap TH$.
Note that on $TM$ the basis dual to $(N,n,e_1, \cdots, e_{n-1})$ is $( -g(n,\cdot),-g(N,\cdot), e^1, \cdots e^{n-1})$ where $e^i(e_j)=\delta_{ij}$, $e^i(n)=e^i(N)=0$, $i,j=1,\cdots, n-1$.
On $TH$ the basis dual to $(n,e_1,\cdots, e_{n-1}))$ is $(-g(N,\cdot),  e^1,\cdots, e^{n-1}))$.

The Ricci tensor is
\[
Ric(X,Y)=\sum_a b^a (R( b_a, X) Y)
\]
where $\{b_a, a=0,1,\cdots,n\}$ is any basis with $\{b^a\}$  the dual basis. In particular, we can let $\{b_a\}=\{N,n,e_1, \cdots, e_{n-1}\}$. We know that since $H$ is totally geodesic, the connection on $H$ is the restriction of $\nabla$, i.e.\ the  connection on $M$. However, we must be careful when calculating the Ricci tensors, since the traces are  different. Actually, this turn out not to be a problem because in the   calculation of $Ric(Y,n)$ there appears the term
\[
b^0 (R( b_0, Y) n)= -g(n, R( N, Y) n)=0
\]
which vanishes by a symmetry  of the Riemann tensor. We conclude that for the calculation of $Ric(Y,n)$ we can use the trace restricted to $TH$ and then
\[
Ric(Y,n) =\textrm{Tr}\{X\mapsto R(X,Y)n\} = \textrm{Tr}\{X\mapsto \dd \omega(X,Y) n\}=\dd \omega(n,Y).
\]
For the last statement of the theorem, observe that for every $X,Y,W\in TH$, by the symmetries of the Riemann curvature
\[
g(W, R(n,X)Y)=-g(X, R(W,Y)n)=-g(X,n)\dd \omega(W,Y)=0,
\]
so being $W\in TH$ arbitrary, we get $R(n,X)Y=\mu(X,Y) n$ where $\mu$ is clearly bilinear. By the first Bianchi identity $R(n,X)Y+R(X,Y)n+R(Y,n)X=0$ that is $\mu(X,Y)+\dd \omega(X,Y)-\mu(Y,X)=0$.
$\square$ \end{proof}



The next result is crucial for the generalization to the non-vacuum case. It can also be found  in the proof of \cite[Prop.\ 4.1]{racz00}. Since the cosmological constant does not play a role in the proof, it is omitted from the statement.

\begin{lemma} \label{ber}
Suppose that $(M,g)$ satisfies the dominant energy condition. Then the following property holds: If $\textrm{Ric}(X,X)=0$ for some lightlike vector $X$, then  $\textrm{Ric}(X,)\propto g(X,\cdot) $ (and similarly replacing the two instances of $\textrm{Ric}$ with $T$ as given by (\ref{fir})).
\end{lemma}

\begin{proof}
It is clear that  the statement holds for $\textrm{Ric}$ iff it holds for $T$, the difference between the two tensors being proportional to $g$.
We can assume that $X$ is future-directed, the statement to be proved being independent of its orientation.
The dominant energy condition is the statement that the map  (\ref{fir}) sends the future causal cone into the past causal cone so that, in particular, for $X$ causal
\begin{equation} \label{car}
g(X,[Ric(X,\cdot)-[\frac{R}{2}+\Lambda]  g(X,\cdot) ]^\sharp)= Ric(X,X)-[\frac{R}{2}+\Lambda]  g(X,X)
\end{equation}
is non-negative and vanishing (by the equality case of the reverse triangle inequality \cite{minguzzi18b}) iff $X$ is lightlike and $Ric(X,\cdot)-[\frac{R}{2}+\Lambda]  g(X,\cdot) \propto g(X,\cdot)$.
We can apply the equality case because, due to the assumptions on $X$, the right-hand side of (\ref{car})
vanishes. We conclude that  $Ric(X,\cdot) \propto g(X,\cdot)$.
$\square$ \end{proof}

The first statement of the next result was known for Killing horizons, e.g.\ \cite[After Prop.\ 4.3.11]{chrusciel20}).

\begin{lemma} \label{viw}
Assume the dominant energy condition.
Then
 $Ric(n,\cdot)\propto g(n,\cdot)$ on $H$, thus $Ric(n,\cdot)\vert_{TH}=0$. Hence we have
 \begin{equation}
 \dd \omega(n,\cdot)=0
 \end{equation}
  and $\mu(n,\cdot)=\mu(\cdot, n)=0$.
\end{lemma}

This means that the 2-form $\dd \omega$, section of $T^*H\otimes T^*H$, passes to the quotient to a section of $(TH/n)^*\otimes (TH/n)^*$, and similarly for $\mu$ and $Ric\vert_{TH\times TH}$. The quotient space $TH/n$ is   endowed with a positive definite metric which is the quotient of $g$.

\begin{proof}
As already observed,  we have $\textrm{Ric}(n,n)=0$ on $H$.
By Lemma \ref{ber} $Ric(n,\cdot)\propto g(n,\cdot)$ and hence $Ric(n,\cdot)\vert_{TH}=0$.



Note that $\mu(n,Y)n=R(n,n)Y=0$ and $\mu(X,n)n= R(n,X)n=\dd \omega(n,X) n=0$.
$\square$ \end{proof}

%
%

\begin{lemma} \label{loq}
For every vector field $X: H \to TH$,
\begin{equation} \label{vhw}
L_X \omega=\dd \omega(X, \cdot)+\dd (\omega(X))
\end{equation}
thus we have under the dominant energy condition (or, more weakly, under $Ric(n,\cdot)\vert_{TH}=0$)
\[
L_n\omega=\dd \kappa, \quad  \textrm{and} \quad (L_X \omega)(n)=\p_n (\omega(X)),
\]
and hence $L_n \dd \omega=0$ which implies that $L_n \mu$ is symmetric.
\end{lemma}

\begin{proof}
It follows from Cartan's magic formula $L_X=i_X \dd + \dd i_X$. The last result follows from $L_n \dd \omega=\dd L_n\omega=\dd \dd \kappa=0$.
$\square$ \end{proof}

We see that the problem of finding a gauge in which $\kappa$ is constant coincides with that of finding a gauge in which $\omega$ is invariant along the flow of $n$. Suppose that we can find a gauge in which $\kappa=-1$. Since the generators of  $H$  are not necessarily closed, and since there is in general no codimension-2 spacelike hypersurface $\Sigma\subset H$ globally transverse to the generators of $H$, we cannot regard $H$ as a principal bundle. If that were the case $-\omega$ would be a connection \cite{kobayashi63} for the bundle $H\to \Sigma$ as it is left invariant by the flow and has value 1 over $n$.

\begin{remark}
The property $Ric(n,\cdot)\vert_{TH}=0$ appearing in Lemmas \ref{ber} and \ref{viw} was used in \cite{rendall98} where, however, it was not observed that it follows from the dominant energy condition. As a consequence, the main result of that work can be improved as follows
\footnote{
There are chances that the dimensionality assumption in the results of \cite{rendall98}, and hence in this result, could be dropped. We did not check this point in detail.}

\begin{theorem}
Suppose that $(M,g)$ is 4-dimensional and that it  satisfies the dominant energy condition. The horizon $H$ collapses with bounded diameter.
\end{theorem}

This result already places some strong constraints on the topology of $H$. For instance, in a four-dimensional spacetime $H$ cannot be of hyperbolic type, see \cite{rendall98} for a complete discussion.
\end{remark}

\section{Ribbon argument and future incompleteness} \label{buq}
The next result is well known \cite[Lemma B1]{friedrich99} \cite{moncrief08}. We include a coordinate independent proof for completeness.
\begin{lemma} \label{hnw}
Let $\tilde g=g\vert_{TH\times TH}$, then
\[
L_n \tilde g=0
\]
\end{lemma}

\begin{proof}
Let $Y,Z:H\to TH$
\begin{align*}
(L_n \tilde g)(Y,Z)&= g(\nabla_n Y, Z)+g( Y, \nabla_n Z)-g(L_n Y, Z)-g(Y,L_nZ)\\
&=g(\nabla_Y n, Z)+g(Y, \nabla_Z n)=\omega(Y)g( n, Z)+\omega(Z) g(Y,  n)=0.
\end{align*}
$\square$ \end{proof}

We now introduce a `horizontal' distribution aimed at splitting  $TH$.

\begin{lemma} \label{iac}
Let $n^*: H \to T^*H$ be a smooth 1-form field such that $n^*(n)=1$, then there is a constant $C>0$ such that at every $p\in H$ and for every $X\in \textrm{ker} \, n^*$,
\[
\vert \omega(X)\vert \le C \sqrt{g(X,X)}.
\]
\end{lemma}

\begin{proof}

Let us introduce a complete Riemannian metric on $H$. It is sufficient to prove the formula for $X\in  \textrm{ker} \, n^*$ normalized with respect to that metric. But $\sqrt{g(X,X)}>0$ for any such $X$ because $g$ is semi-positive definite on $H$, the degenerate direction being $n$, while $X$ at each point belongs to a hyperplane  transverse to $n$. Thus the bundle  of normalized $X$ in $\textrm{ker} \, n^*$ is compact and $\vert \omega(X)\vert/\sqrt{g(X,X)}$ is finite and continuous on the bundle, thus upper bounded.
$\square$ \end{proof}

As in \cite{reiris21} we construct ribbons closed by $n^*$-horizontal curves that we now introduce.

A $C^1$ curve on $H$ is said to be $n^*$-horizontal if its tangent belongs to $\textrm{ker} n^*$. Two $n^*$-horizontal curves  $\sigma_{0,1}:[0,1] \to H$ are homotopically related if, letting $\varphi: H\times \mathbb{R}\to H$ denote the flow of $n$, we have $\sigma_1(r)=\varphi(\sigma_0(r), s(r))$ for some continuous function $s(r)>0$.
%
%

\begin{lemma} \label{vig}
For every $n^*$-horizontal curve $\sigma$ there is a constant $B>0$ such that
\[
\vert \int_{\tilde\sigma}  \omega \vert \le B ,
\]
for every $n^*$-horizontal curve $\tilde\sigma$ in the same homotopy class of $\sigma$.
\end{lemma}

\begin{proof}
We know that $\vert \int_\sigma  \omega \vert \le C \int \sqrt{g(\sigma',\sigma')} \dd r$. The integrals on the right-hand side coincide for any two homotopic $\sigma_0$ and $\sigma_1$ because $\sigma_1'=\varphi_*(\sigma_0' , s(r))+ \textrm{term prop.\ to } n$, thus
\[
g(\sigma_1',\sigma_1')=(\varphi^*g)(\sigma_0',\sigma_0')=g(\sigma_0',\sigma_0')
\]
where we used $L_n \tilde g=0$.
$\square$ \end{proof}

The 1-form field $n^*$ induces a splitting of $TH$.
In what follows for $X\in TH$ we might  use the notation

\begin{equation} \label{rgf}
X=X^\perp+\lambda n,
\end{equation}
 where $X^\perp$ is $n^*$-horizontal and $\lambda=n^*(X) \in \mathbb{R}$.
 We might also write $\lambda(X)$ in place of $n^*(X)$.
\begin{lemma} \label{nwx}
On $H$ there is a constant $K>0$ such that for every  vector $X\in TH$, we have the affine bound
\[
\vert n^*(d\varphi_s(X))\vert\le K \sqrt{g(X,X)} \, s +\vert n^*(X)\vert.
\]
\end{lemma}

This type of bound will be fundamental in our study and distinguishes our treatment from those of previous references.

\begin{proof}
Let us consider the integral curve of $n$ starting from $p$ and ending at $\varphi_s(p)$, $\sigma: [0,s] \to H$, $r\mapsto \varphi_r(p)$. Notice that $\sigma$ can have self intersections. Let $W(r)\in T_{\varphi_r(p)}H$ be the vector  on the image of $\sigma$ obtained by flowing $X$ with the 1-parameter group of diffeomorphisms generated by  $n$, i.e.\  $W(r)=\dd \varphi_r(X)$.
Observe that all over $H$, $0=L_n(1)=L_n (n^*(n))=(L_n n^*)(n)$, thus on the image of $\sigma$ we have
\[
\frac{\dd }{\dd r} n^*(W(r))=(L_n n^*)(W(r))=(L_n n^*)(W(r)^\perp).
\]
Observe that locally the first step of the above calculation can be performed in a small subinterval $I\subset [0,s]$, $r\in I$, such that $\sigma\vert_I$ has no self intersections. Then $W$ can be regarded as a field on $\sigma(I)$ and hence we can replace  $\frac{\dd }{\dd r}$ with $L_n$ and use $L_n W=0$.

The symmetric bilinear form $\tilde g$ is positive definite once restricted to the horizontal bundle $\ker n^*\subset TH$, as the latter is transversal to $n$ at every point.
As the projective bundle of the horizontal bundle  is compact, the function   $Z\mapsto (L_n n^*)(Z)/\sqrt{g(Z,Z)}$ is continuous and hence bounded on it, thus there is $K>0$ such that for every $n^*$-horizontal vector $Z$,  $\vert(L_n n^*)(Z)\vert\le K\sqrt{g(Z,Z)}$. Finally, since $L_n \tilde g=0$,
\begin{align*}
\vert (L_n n^*)(W(r)^\perp) \vert\le K\sqrt{g(W(r)^\perp,W(r)^\perp)} &=K \sqrt{g(d\varphi_r(X),d\varphi_r(X))}\\
&= K\sqrt{g(X,X)}.
\end{align*}
By integrating in $r$, in the interval $[0,s]$, we obtain  a bound on $\vert n^*(d\varphi_s(X))-n^*(X)\vert \ge \big\vert \vert n^*(d\varphi_s(X))\vert-\vert n^*(X)\vert \big\vert$.
$\square$ \end{proof}

\begin{itemize}
\item[($\star\star$)] {\em In what follows we shall assume the dominant energy condition without further notice (i.e.\ the assumption of Lemma \ref{viw}) or the weaker condition $\textrm{Ric}(n,\cdot)\vert_{TH}=0$.}
\end{itemize}

The importance of the condition $\textrm{Ric}(n,\cdot)\vert_{TH}=0$ was recognized in \cite[Thm.\ 1.2]{rendall98} \cite[Assumption 2.1]{petersen19}.
The fact that it follows from the dominant energy condition was  noted in  \cite[Prop.\ 4.1]{racz00} \cite[Rem.\ 1.15]{petersen18b}.

A {\em ribbon} is the homotopy  of two  homotopically related $n^*$-horizontal curves $\sigma_0$ and $\sigma_1$ (hence $\sigma_1(r)=\varphi(\sigma_0(r), s(r))$ for some continuous function $s(r)>0$). If the ribbon is injective (no self-intersection) its image is graphically bounded by four  curves, the horizontal sides $\sigma_0$ and $\sigma_1$ and two generator segments $\gamma_0$, $\gamma_1$ starting from $\sigma_0(0)$ and $\sigma_0(1)$ respectively, see Fig.\ \ref{frg}.


\begin{figure}[ht]
\centering
\includegraphics[width=4cm]{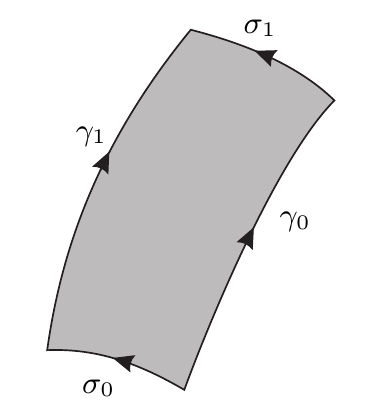}
\caption{The curves $\sigma_0$ and $\sigma_1$ are $n^*$-horizontal.  The integral $\int_{\sigma_1} \omega$ is bounded by a constant which is independent on how elongated  the ribbon is in the longitudinal direction.
The future endpoints of $\gamma_0$ and $\gamma_1$ correspond to values of the integral parameter of $n$ that do not necessarily coincide, however there is a one-to-one correspondence between them so one goes to infinity iff the other goes to infinity (cf. Appendix).
} \label{frg}

\end{figure}

Stokes theorem applied to $\dd \omega$  (which vanishes on the tangent bundle to the ribbon) gives
\begin{equation} \label{sta}
\int_{\sigma_0} \omega+\int_{\gamma_1} \omega-\int_{\sigma_1} \omega-\int_{\gamma_0}\omega=0.
\end{equation}
We are not considering here the possibility that the ribbon self intersects. Still one can obtain Eq.\ (\ref{sta}) also in this case, by  splitting the ribbon in several shorter, non-self intersecting ribbons for which the previous equation in display holds. Then summing all the equations so obtained one gets the equation for the original ribbon (the interior horizontal segments give opposite contributions that cancel out). In other words, the global validity of the equation follows from its local validity.

Alternatively, and more precisely, we can integrate the pullback of $\dd \omega$ by the homotopy map \cite{reiris21,moncrief08}:
Let $A=\{(r,u): 0\le u \le s(r), \ r\in [0,1]\}\subset \mathbb{R}^2$.
Let $\chi: A\to H$ be the map  (the map $\chi$ is $C^1$, see the Appendix)
\[
\chi: (r,u)\mapsto \varphi(\sigma_0(r), u)
\]
where $\varphi$ is the flow of $n$. Observe that for fixed $r$, $\chi(r,\cdot): [0,s(r)]\to H$ maps to an integral curve of $n$, thus $\chi_*(\p_u)\propto n$.

Now we apply Stokes theorem to $A$ and the form $\chi^*\omega$
\[
\int_A \dd  \chi^* \omega=\int_{\p A} \chi^*\omega
\]
observe that $\dd  \chi^* \omega=\chi^*\dd \omega=0$ because $A$ is a subset of $\mathbb{R}^2$ and
\[
(\chi^*\dd \omega)(\cdot, \p_u)=\dd \omega(\cdot, \chi_*(\p_u)) \propto \dd\omega(\cdot, n)=0.
\]
Thus
\[
0=\int_{\p A} \chi^*\omega
\]
which is the Equation (\ref{sta}). By using our Lemma \ref{vig} we arrive at
\begin{equation}
\vert \int_{\gamma_1} \omega-\int_{\gamma_0}\omega \vert\le 2B
\end{equation}
where the constant $B(\sigma_0)$ does not depend on  how extended  the ribbon is in the longitudinal direction.
For a discussion of the existence and extendibility of ribbons (homotopies) we refer the reader to the Appendix.

\begin{proposition} \label{vos}
The validity of the property ``the integral $\int_{\gamma_p([0,s))} \omega$ over the future-directed generator  starting from $p$ converges to  $-\infty$'' (resp. $+\infty$, is upper bounded, is lower bounded) does not depend on the generator $\gamma_p$ considered.
\end{proposition}

Thus, if for one generator the integral is not lower (upper) bounded, the same is true for every generator.
\begin{proof}
For every  $p\in H$  we can find  a cylindrical neighborhood $C$ whose quotient under the flow of $n$ is a disc of radius $\rho$ in the quotient metric (which, for sufficiently small cylindrical neighborhood,  is well defined by Lemma \ref{hnw}), while the height is some $\delta$ in the parameter of $n$ (cf. Figure \ref{xgq}).

\begin{figure}[ht]
\centering
\includegraphics[width=9cm]{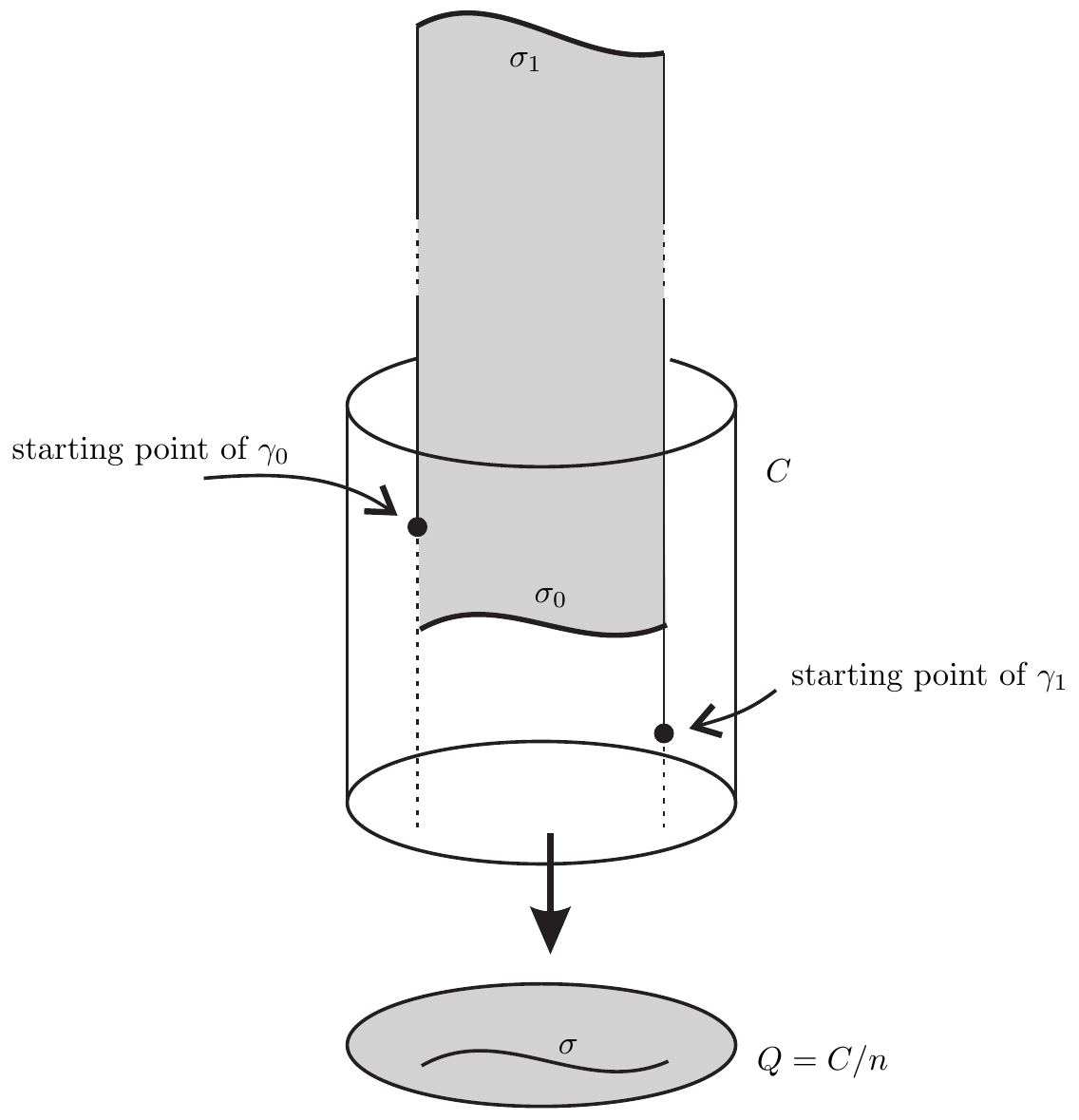}
\caption{Any two generators $\gamma_0$ and $\gamma_1$ starting from the cylinder, and suitably prolonged, can be  connected by a curve $\sigma_0$ which is the horizontal lift of a geodesic $\sigma$ living on the quotient space $Q$.} \label{xgq}
\end{figure}

 We know that for sufficiently small $\rho$ the quotient disc is actually convex  \cite[Sec.\ 4.7]{docarmo16}, that is, any two points on the quotient can be connected by a geodesic of the quotient metric.
Let us consider the  $n^*$-horizontal lifts of such geodesics, namely the $n^*$-horizontal curves projecting on them
(See the Appendix for more details on this concept). Their existence shows that any two generators starting from the cylinder can be joined by a horizontal curve $\sigma_0$ provided the starting point of the generator is suitably translated on the generator. Consider an arbitrarily elongated ribbon with starting  horizontal curve $\sigma_0$. The image of the ribbon and the closing horizontal curve $\sigma_1$ need not stay inside the  cylinder (the quotient construction is used to find $\sigma_0$ while for the existence of the ribbon we refer the reader to the Appendix).

 Now observe that $\int \omega$ evaluated over any generator segment of the cylinder   is bounded by some a priori constant $R>0$ (because the segments are determined by their parameter length and by their starting point, and this pair belongs to a compact set $\bar C\times [0,\delta]$ while the function $\int \omega$ depends continuously on these variables), thus for any two generator segments starting from the cylinder and ending outside it we have
\begin{equation} \label{vkq}
\vert \int_{\gamma_1} \omega-\int_{\gamma_0}\omega \vert\le 2B +2R ,
\end{equation}
where it is understood that the two final endpoints of the generators are connected by the horizontal curve $\sigma_1$.

The constants $2B$ and $2R$ do not depend on how long the ribbon is chosen to be, thus we conclude that  $\int_{\gamma_0} \omega$ converges to a  $\pm$ infinite value if and only if $\int_{\gamma_1} \omega$ converges to a $\pm$ infinite value. Similarly in the bounded cases.
Thus the starting points that lead to an integral $\int \omega$ which is of a given type (e.g.\ converges to $+\infty$) form open sets, as do those of the negative type (resp.\ does not converge to $+\infty$)  and so connected components of the horizon. As the horizon is connected we conclude that on the horizon we can only have one of the behaviors.
$\square$ \end{proof}

We are going to show that if $H$ admits a future incomplete generator then all geodesic generators are future incomplete (Cor.\ \ref{cbv}). Our proof makes use of the next important observation.

\begin{lemma} \label{iif}
Let $f: [0,\infty)\to \mathbb{R}$ be a continuous bounded function. If $g(t):=\exp ( \int^t_0 f(s) \dd s)$ is integrable then $\int_0^{\infty} f(t) \dd t=-\infty$.
\end{lemma}

\begin{proof}
By a result due to Lesigne \cite{lesigne10} for almost every $x>0$ we have
\[
\lim_{n\to +\infty}  g(nx) =0,
\]
hence
for a choice of such $x$ we have since $g$ is positive
\[
\lim_{n\to +\infty} \int^{nx}_0 f= \lim_{n\to +\infty}  \log g(nx) =-\infty
\]
For any $\rho>0$ we can write $\rho=[\rho]_x x+r(\rho)$ where $[\rho]_x$ is the integer such that $r\in [0,x)$. Of course, for $\rho \to \infty$ we have $[\rho]_x\to \infty$. Note that
\[
\vert \int^{\rho}_ {[\rho]_x} f(s)\dd s \vert \le K x
\]
where $K>0$ is the bound for $f$, i.e.\ $\vert f\vert \le K$. Thus
\[
\lim_{\rho \to +\infty} \int^\rho_0 f= \lim_{\rho \to +\infty} ( \int^{[\rho]_x}_0 f+ \int^{\rho}_ {[\rho]_x} f )\le \lim_{\rho \to +\infty} ( \int^{[\rho]_x}_0 f+Kx )=-\infty.
\]

$\square$ \end{proof}

\begin{corollary}
If the horizon  admits one future incomplete generator, then \[\int_{\gamma([0,\infty))} \omega=-\infty\] over every generator.
\end{corollary}

Remember that over a generator $\int_{\gamma([0,s))} \omega=\int_0^s \kappa(\gamma(s)) \dd s$ where $s$ is the integral parameter of $n$. For shortness, we might denote the argument ``$\kappa(s)$''.
\begin{proof}
The function $\kappa$ is continuous on the horizon, which is compact, thus it is bounded. Incompleteness of the generator $\gamma$ reads $\int_0^\infty (\exp \int_0^t \kappa(s) \dd s) \dd t< \infty$, i.e.\ $g(t):=\exp \int_0^t \kappa(s) \dd s$ is integrable.
By Lemma \ref{iif}
$\int_\gamma \omega=-\infty$ over the incomplete generator, thus the same equation holds for any generator by Prop.\ \ref{vos}.
$\square$ \end{proof}

\begin{proposition} \label{viu}
Assume $H$ is such that $\int_{\gamma([0,\infty))} \omega=-\infty$ over every generator (e.g.\ because it admits a future incomplete generator). There is a constant $C>0$ such that for every generator parametrized with the integral parameter of $n$
(and regardless of the zero point of the parametrization)
\begin{equation} \label{vid}
\int_0^\rho \kappa(s) \dd s<0
\end{equation}
 for every $\rho \ge C$.
\end{proposition}

In other words, the integral of surface gravity over a segment of generator of $n$-parametrization-length no smaller than $C$ is negative.

\begin{proof}
It is sufficient to prove the result for generators starting from a suitable neighborhood $U_p$ of an arbitrary point $p$. By a standard compactness argument, calling $C_p$ the constant, passing to a finite subcovering covering $\{U_{p_i}\}$, $C=\max \{C_{p_i}\}$ provides the constant for the whole horizon.

So let $p\in H$, we know that there is a cylindrical neighborhood $U_p$ of height $\delta$ in the $n$-parametrization such that Eq.\ (\ref{vkq}) holds (actually, taking $U_p$ slightly smaller we can let that equation hold in the closure of $U_p$).
Using notation as in the proof of Prop.\ \ref{vos} (including the definitions of $\delta$ and $R$), if the stronger inequality
\begin{equation}  \label{vbd}
\int_0^\rho \kappa(s) \dd s<-R
\end{equation}
holds for a generator starting from $q\in U_p$ and for $\rho\ge \tilde C$ then (\ref{vid}) is satisfied for any starting point $q'\in U_p$ on the same generator segment of $U_p$, it is sufficient to define $C= \tilde C+\delta$. This means that it is sufficient to prove the inequality (\ref{vbd})  for one starting point over each generator segment of the cylinder. If $\gamma_0$ is the geodesic starting from $p=:p_0$,  and $\gamma_1$ is another geodesic starting from a different generator of the cylinder, we pass to the quotient, consider their projections $\bar p_0$ and $\bar p_1$ and redefine the starting point  $p_1$ to stay in the horizontal lift of the unique geodesic connecting $\bar p_0$ to $\bar p_1$.  The union of the points so obtained gives a set $S$ on $U_p$ (its boundary is contained in $\p U_p$). Over every generator starting from $\bar S$ we consider the $n$-parameter $z$ which vanishes on $\bar S$.

For $q\in S$ let $z \mapsto \phi(q,z)$ be the horizontal transport map of the $n$-parameter $z$ over the $\bar p$-fiber to the $n$-parameter $z'=\phi(q,z)$ over the $\bar q$-fiber obtained by horizontally transporting along the geodesic connecting $\bar p$ to $\bar q$. This map is continuous and increasing in $z$ (no two horizontal lifts of the same curve can intersect). Let $T$ be such that  for every $t\ge T$, $ \int_0^t \kappa(s) \dd s<-2B-3R$ over the generator starting from $p$. Let $\tilde C$ be the maximum of $q \mapsto \phi(q,T)$ over $\bar S$.
Then (\ref{vkq}) reads for $\rho \ge \tilde C$
\begin{equation}
\vert \int_{0,\gamma_1}^{\rho} \kappa(s) \dd s - \int_{0,\gamma_0}^{t} \kappa(s) \dd s \vert \le 2B+2R
\end{equation}
where the ribbon projects to the radial geodesic starting from $\bar p$, thus
as $\rho\ge \tilde C$ we have $t\ge T$ and hence we get that (\ref{vbd}) is valid over every generator starting from $S$.
Finally, let $C=\tilde C+\delta$.
$\square$ \end{proof}

\begin{theorem} \label{bhq}
Assume $H$ is such that $\int_{\gamma([0,\infty))} \omega=-\infty$ over every generator (e.g.\ because it admits a future incomplete generator).
There is a smooth function $f:H\to \mathbb{R}$ such that $n'=e^f n$ has (smooth) surface gravity $\kappa'<0$ all over $H$.
\end{theorem}

\begin{proof}
We define $f$ through
\[
e^{f(p)}= \int_0^C e^{ \int_{\varphi(p, [0,s])} \omega} \dd s ,
\]
where $C>0$ is as in Prop.\ \ref{viu}.  It is clear that the function is smooth, thanks to the fact that it is defined via integration in a compact interval.
We want to calculate $\p_n f(p)$. Let $s \mapsto x(s)$ be the integral curve of $n$ starting from $p$. Similarly to Eq.\ (\ref{bud}) we have
\[
e^{f(x(\tau))}=\int_\tau^{\tau+C}e^{\int_\tau^\rho \kappa(x(s))\dd s} \dd \rho.
\]
Differentiating with respect to $\tau$ (here we can safely differentiate under the integral sign because the domain of integration is compact) and setting $\tau=0$ we get
\[
e^f \p_n f=-1 +e^{\int_0^{C} \kappa(x(s))\dd s}-\kappa e^f
\]
or, by Eq.\ (\ref{ok5}), $\kappa'=e^f(\kappa +\p_n f)=-1 +e^{\int_0^{C} \kappa(x(s))\dd s}<0$.
$\square$ \end{proof}

\begin{corollary} \label{cbv}
Assume $H$ admits a future incomplete generator.
All geodesic generators are future incomplete.
\end{corollary}

\begin{proof}
We know that $n$ can be chosen so that $\kappa<0$.
As $\kappa$ is continuous and negative, by the compactness of $H$ there is some $K<0$, such that $k\le K<0$. Denoting with $x(s)$ the integral curve of $n$ starting from $p$, the  geodesic $\gamma$ starting from $p$ with tangent $\dot \gamma=n$ has future affine length
\begin{equation}
\Lambda(p)=\int_0^\infty e^{\int_0^\rho \kappa(x(s))\dd s} \dd \rho\le \int_0^\infty e^{ K \rho} \dd \rho \le \frac{1}{-K}<\infty.
\end{equation}
$\square$ \end{proof}

\section{The differential and Hessian of $\Lambda$}

In this section our goal is to prove that $\Lambda$ is $C^2$, and  that its Hessian admits an expression that, in vacuum, allows for a bootstrap argument (Cor.\ \ref{vja}). We start by proving that it is $C^1$.

\begin{theorem}
Assume $H$ admits a future incomplete generator.  The function $\Lambda$ is $C^1$.
\end{theorem}

\begin{proof}
From Eq.\ (\ref{ok4}) we know that this property is independent of the choice of $n$ so we can choose $n$ such that $\kappa\le K<0$, and hence $\Lambda\le 1/(-K)$.

Let us calculate $\p_X \Lambda (p)$, $X\in T_pH$. Let us extend $X$ by pushing it forward with the flow of $n$ so as to obtain a map $s\mapsto X(s)=\dd \varphi_s(X)$, (elsewhere, to save space, we might also denote the pushforward with $X^s$) which defines a vector at $x(s)$, where $s\mapsto x(s)$ is the integral curve of $n$ starting from $p$. This is not really a vector field defined over a neighborhood of the generator passing from $p$ (as the generator can accumulate on itself and the vector field would be multi-valued) but with some abuse of notations we can write $[X,n]=0$, an equation which is correct if interpreted locally so as to get single valuedness.

The function $\Lambda$ is
\[
\Lambda(p)=\int_0^\infty e^{ \int_{\varphi(p, [0,s])} \omega} \dd s
\]
Now we are going to assume that we can switch the derivative and the integral. This key first step will be justified later when we shall show that the argument of the integral so obtained is continuous and integrable, a fact that allows one to apply the dominated convergence theorem to justify the first step.

Observe that by varying $p$ in direction $X$ we are really varying the generator. This is done with the vector field $X$ that locally is such that $[X,n]=0$ an equation which tells us that $X$ preserves the $n$-parametrization $\rho$ and  the related measure $\dd \rho$
\begin{align} \label{vod}
\p_X \Lambda (p)=\int_0^\infty [ \int_{\varphi(p, [0,\rho])}  L_X \omega ]   \, e^{ \int_{\varphi(p, [0,\rho])} \omega} \dd \rho
\end{align}
Now we use Eq.\ (\ref{vhw}) $L_X \omega=\dd \omega(X, \cdot)+\dd (\omega(X))$ observing that the first term on the right-hand side vanishes when integrated as $\dd \omega$ is annihilated by $n$ (cf.\ Lemma \ref{viw}). We arrive at
\begin{align}
\p_X \Lambda (p)=\int_0^\infty [ \omega(X(\rho))-\omega(X(0))] \,   e^{ \int_{\varphi(p, [0,\rho])} \omega} \dd \rho \label{vob} \\
=\int_0^\infty [ \omega(X(\rho))-\omega(X(0))]  \,  e^{ \int_0^\rho \kappa(x(s)) \dd s} \dd \rho.
\end{align}
Linearity of this expression in $X$ is clear. If we can show that the argument is integrable we justify two facts: (a) that the switching first step was justified, and (b) that this expression is actually continuous in $p$ and hence provides the continuous differential of $\Lambda$. Both follow from the dominated convergence theorem.

Since we already know that $\kappa\le K<0$ it is sufficient to show that  $[ \omega(X(\rho))-\omega(X(0))]$ is bounded by a polynomial in $\rho$.
Let us introduce the  1-form field $n^*$ as in Lemma \ref{iac} and let us split $X=X^\perp+\lambda n$ with  $X^\perp \in \ker n^*$. From the same Lemma and from $L_n X=0$, $L_n g=0$
\[
\vert \omega(X^\perp(\rho))\vert \le C \sqrt{g(X^\perp(\rho),X^\perp(\rho))}=  C \sqrt{g(X(\rho),X(\rho))}=C \sqrt{g(X(0),X(0)}
\]
By Lemma \ref{nwx} $\lambda$ is affinely bounded.
We have $\vert \omega(\lambda n)\vert =\vert \lambda \kappa\vert$ and since $\vert \kappa\vert$ is continuous on the compact $H$ it is also bounded, which proves that  $\omega(\lambda n)$ is affinely bounded. We conclude that $\vert \omega(X)\vert$ is bounded by an affine expression in $\rho$, which ends the proof.
$\square$ \end{proof}

The next result will be used to show that the condition $L_n\textrm{Ric}\vert_{TH}=0$, that we shall use in the bootstrap argument, is actually independent of the choice of vector field $n$ and so represents a well defined property of the horizon.

\begin{lemma} \label{kke}
Let $\sigma: H\to T^*H\times T^*H$ be a bilinear form on the horizon such that $\sigma(\cdot, n)=\sigma(n, \cdot)=0$, then $L_n\sigma$ has the same property. Moreover, $(L_{e^fn}\sigma)=e^f(L_{n}\sigma)$.
\end{lemma}

\begin{proof}
For any vector field $X: H\to TH$
\[
(L_n\sigma)(X,n)=L_n[\sigma(X,n)]-\sigma(L_nX,n)-\sigma(X, L_n n)=0.
\]
For any vector fields $X, Y : H\to TH$
\begin{align*}
(L_{e^fn}\sigma)(X,Y)&=L_{e^fn}[\sigma(X,Y)]-\sigma(L_{e^fn}X,Y)-\sigma(X, L_{e^fn} Y)\\
&=e^f[L_n[\sigma(X,Y)]-\sigma(L_{n}X,Y)-\sigma(X, L_{n} Y)]\\
&\quad+e^f X(f)\sigma(n,Y)+e^f Y(f) \sigma(X,n)=e^f(L_n\sigma)(X,Y).
\end{align*}
$\square$ \end{proof}

\begin{theorem} \label{bjz}
Assume $H$ admits a future incomplete generator. The bilinear form $\eta: H\to T^*H\otimes T^*H$ on the horizon, given for $X,Y\in TH$ by
\begin{equation} \label{kxc}
\eta(X,Y):=\mu(Y,X)+i_X \nabla_Y\omega+\omega(X) \omega(Y),
\end{equation}
is symmetric and such that $(L_n \nabla)(Y)X=\eta(X,Y)n $.
 Moreover, the function $\Lambda$ is $C^2$ and it  satisfies
\begin{align} \label{kkv}
\begin{split}
&\textrm{Hess}\, \Lambda+ \omega\otimes \dd \Lambda+\dd \Lambda\otimes \omega+\Lambda \eta =B,
\end{split}
\end{align}
where the differential operator $T$ on the left-hand side is symmetric and under a gauge transformation $n'=e^fn$, $\Lambda'=e^{-f}\Lambda$, it satisfies (note that $\omega$ and hence $\eta$ change) $T' (\Lambda')=e^{-f} T \Lambda$, while
\begin{equation} \label{bcw}
B(p)(X,Y):=-\int_0^\infty   [\varphi_\rho^* \mu- \mu](Y,X) \, e^{ \int_{\varphi(p, [0,\rho])} \omega} \dd \rho,
\end{equation}
is a $C^0$ symmetric bilinear form.

\end{theorem}

\begin{proof}
We start proving that it is $C^2$.
We calculate the Hessian
\[
\textrm{Hess} \,\Lambda(Y,X)=\p_Y \p_X \Lambda-\p_{\nabla_Y X} \Lambda
\]
For the former term on the right-hand side we differentiate  Eq.\ (\ref{vod}) by switching the order of differentiation and integration. Again, we shall prove that this step is allowed by proving that the differentiated argument is integrable. The latter term is  given directly in terms Eq.\ (\ref{vob})
which we already proved to be valid. Instead of showing the integrability of the argument for $\p_Y \p_X \Lambda$ we show directly that property for $\textrm{Hess} \,\Lambda(Y,X)$ as their difference is the argument of the integral expression for $\p_{\nabla_Y X} \Lambda$ which we already know to be integrable (in fact affine times a
converging exponential).
\begin{align}
\begin{split}
\p_Y \p_X \Lambda (p)=\int_0^\infty \Big\{&[ L_Y \int_{\varphi(p, [0,\rho])}  L_X \omega ] \\
&+ [ \int_{\varphi(p, [0,\rho])}  L_X \omega ]  [ \int_{\varphi(p, [0,\rho])}  L_Y \omega ]  \Big\}  \, e^{ \int_{\varphi(p, [0,\rho])} \omega} \dd \rho
\end{split}
\end{align}
\begin{align}
\p_{\nabla_Y X} (p)=\int_0^\infty   [ \int_{\varphi(p, [0,\rho])}  L_{\nabla_Y X} \omega ]    \, e^{ \int_{\varphi(p, [0,\rho])} \omega} \dd \rho
\end{align}
\begin{equation}
\begin{split}
\textrm{Hess} \Lambda(Y,X) (p)=\int_0^\infty \Big\{&[ L_Y ( \omega(X))(\rho)-L_Y(\omega(X))(0) ] \\
&-[\omega( (\nabla_YX)^
\rho)-\omega(\nabla_YX(0))]  \\
&\!\!\!\!\!\!\!\!\!\!\!\!\!\!\!\!\!\!\!\!\!\!\!\!\!\!\!\!+ [ \omega(X(\rho))-\omega(X(0))]  [ \omega(Y(\rho))-\omega(Y(0))]  \Big\}  \, e^{ \int_{\varphi(p, [0,\rho])} \omega} \dd \rho
\end{split}
\end{equation}
To avoid ambiguities we have denoted $(\nabla_YX)^\rho$ the push forward with the flow of $n$ of $\nabla_YX(0)$. This is not necessarily $\nabla_YX(\rho)$ (remember that $X,Y$ are invariant under the flow of $n$, but to stress the fact that they have been obtained pushing $X(0)$, $Y(0)$ with the flow of $n$,   we might write $X^\rho$, $Y^\rho$ and hence denote $\nabla_YX(\rho)\to \nabla_{Y^\rho}X^\rho$).
Let
\[
\nabla \omega(X,Y):=i_X \nabla_Y \omega=(\nabla_Y\omega)(X),
\]
so that, $\dd \omega(X,Y)=\nabla \omega(Y,X)-\nabla \omega(X,Y)$, and let us rewrite the Hessian as follows
\begin{align}
\begin{split}
\textrm{Hess} \Lambda(Y,X) (p)=\int_0^\infty \Big\{&[ \nabla \omega (X(\rho), Y(\rho))- \nabla \omega (X(0), Y(0)) ]  \\
&\qquad+\omega\left(  \nabla_{Y^\rho}X^\rho- (\nabla_YX)^\rho \right)  \\
&\!\!\!\!\!\!\!\!\!\!\!\!\!\!\!\!\!\!\!\!\!\!\!\!\!\!\!+ [ \omega(X(\rho))-\omega(X(0))]  [ \omega(Y(\rho))-\omega(Y(0))]  \Big\}  \, e^{ \int_{\varphi(p, [0,\rho])} \omega} \dd \rho
\end{split}
\end{align}
We study the terms on the three lines separately.
\begin{enumerate}
\item    The first term is  bounded by a quadratic polynomial in $\rho$, as can be seen splitting $X=X^\perp+\lambda_X n$ and $Y=Y^\perp+\lambda_Y n$ and arguing again as for the $C^1$ case. That is, the compactness of the projective bundle of $n^*$-horizontal vectors implies the existence of a constant $F>0$ such that
\begin{align*}
\nabla \omega (X^\perp(\rho), Y^\perp(\rho))&\le F\sqrt{g(X^\perp(\rho),X^\perp(\rho))}\sqrt{g(Y^\perp(\rho),Y^\perp(\rho))}\\
& \quad = F\sqrt{g(X(\rho),X(\rho))}\sqrt{g(Y(\rho),Y(\rho))}\\
&\quad = F\sqrt{g(X(0),X(0))}\sqrt{g(Y(0),Y(0))}.
\end{align*}
The terms that are not bounded by a constant are those that come from terms in $\lambda_X$ or $\lambda_Y$ which are rather bounded by affine expressions in $\rho$.
For instance, in order to treat the term $\omega (X^\perp(\rho), \lambda_Y n )$ we observe that by the compactness of the projective bundle of $n^*$-horizontal vectors there is  a constant $G>0$ such that
\begin{align*}
\nabla \omega (X^\perp(\rho), n)&\le G\sqrt{g(X^\perp(\rho),X^\perp(\rho))}\\
& \quad = G\sqrt{g(X(\rho),X(\rho))}\\
&\quad = G\sqrt{g(X(0),X(0))}.
\end{align*}
Thus up to a proportionality  constant the bound for that term is the same as that on $\lambda_X$ which is affine in $\rho$.
\item Let us study the second term.
For $Y,X\in TH$ we have, by a standard formula on the Lie derivative of a connection \cite[Eq.\ (2.23)]{yano55}
\begin{equation}
(L_n \nabla)(Y)X=R(n,Y)X+i_X \nabla_Y  \nabla n=R(n,Y)X+( \nabla \omega(X,Y)+\omega(X) \omega(Y)) n
\end{equation}
By Lemma \ref{vie} we have
\begin{equation} \label{lie}
(L_n \nabla)(Y)X=\eta(X,Y)n,
\end{equation}
where
\begin{equation}
\eta(X,Y)=\mu(Y,X)+\nabla \omega(X,Y)+\omega(X) \omega(Y),
\end{equation}
is a smooth bilinear form on the horizon. It is symmetric, indeed from Lemma \ref{vie}
\begin{align*}
\eta(X,Y)-\eta(Y,X)&=\mu(Y,X)-\mu(X,Y)+\nabla \omega(X,Y)-\nabla \omega(Y,X)\\
&=-\dd \omega(Y,X)+\dd \omega(Y,X)=0.
\end{align*}
Now notice that
\begin{align*}
&L_n[\nabla_{Y^\rho} X^\rho -(\nabla_{Y} X)^\rho- (\int_0^\rho \eta(X(s),Y(s))\dd s) n]\\
&=[\eta(X(\rho),Y(\rho))-\eta(X(\rho),Y(\rho))]n=0
\end{align*}
and since the vector in square brackets vanishes at $\rho=0$, we have actually that it vanishes for every $\rho$, thus
\begin{equation} \label{vqq}
\nabla_{Y^\rho} X^\rho -(\nabla_{Y} X)^\rho=[\int_0^\rho \eta(X(s),Y(s))\dd s]\, n.
\end{equation}
%
%
We have
 \[
\omega\left(  \nabla_{Y^\rho}X^\rho- (\nabla_YX)^\rho \right) =  \kappa(x(\rho))   \int_0^\rho  \eta(X(s),Y(s))\dd s.
 \]
By compactness of $H$ and continuity of $\eta$ we can easily bound the various terms that one gets splitting $X$ and $Y$.
More precisely, with the usual splitting  it follows that $\eta(X(s),Y(s))$ is bounded by a quadratic polynomial in $s$ and so the term that we are studying is bounded by a cubic polynomial in $\rho$.
\item The third term is bounded by a  quadratic polynomial in $\rho$, as each factor is bounded by an affine expression in $\rho$.
\end{enumerate}

In conclusion, the integral argument is bounded by a  cubic polynomial in $\rho$ and so it is integrable.

We observe that all three terms are separately symmetric in $(X,Y)$. The symmetry of the first term follows expressing it in terms of $\mu$ and using the symmetry of $L_n\mu$ (see Lemma \ref{loq}). The symmetry of the second term follows from the fact that the commutators of the push forwards is the push forward of the commutators. The symmetry of the last term is obvious.

We can actually simplify the expression for the Hessian further. First observe that
\begin{align*}
&[ \omega(X(\rho))-\omega(X(0))]  [ \omega(Y(\rho))-\omega(Y(0))]\\
 =&-[ \omega(X(\rho))-\omega(X(0))] \omega(Y(0))-\omega(X(0))  [ \omega(Y(\rho))-\omega(Y(0))]\\
&-\omega(X(0))\omega(Y(0))+ \omega(X(\rho)) \omega(Y(\rho))
\end{align*}
Thus by using the expression for $\dd \Lambda$, Eq.\ (\ref{vob}) we get
that the third term can be written
\begin{align*}
-\omega(Y) \p_X \Lambda-\omega(X)\p_Y \Lambda-\omega(X)\omega(Y)\Lambda+\int_0^\infty \omega(X(\rho)) \omega(Y(\rho))  \, e^{ \int_{\varphi(p, [0,\rho])} \omega} \dd \rho
\end{align*}
The first term can be written
\begin{align*}
-\nabla\omega(X,Y)  \Lambda+\int_0^\infty \nabla\omega(X(\rho),Y(\rho))  \, e^{ \int_{\varphi(p, [0,\rho])} \omega} \dd \rho
\end{align*}
The second term can be written
\begin{align*}
&\int_0^\infty   \kappa(x(\rho))   \int_0^\rho  \eta(X(s),Y(s)) \dd s\, e^{ \int_{\varphi(p, [0,\rho])} \omega} \dd \rho\\
=& \int_0^\infty   \kappa(x(\rho))   \int_0^\rho  [\nabla \omega(X(s),Y(s))+\omega(X(s)) \omega(Y(s))]\dd s\, e^{ \int_{\varphi(p, [0,\rho])} \omega} \dd \rho\\
&+\int_0^\infty   \kappa(x(\rho))   \int_0^\rho  \mu(Y(s),X(s)) \dd s\, e^{ \int_{\varphi(p, [0,\rho])} \omega} \dd \rho
\end{align*}
Summing the three terms and recalling that  $\int_{\varphi(p, [0,\rho])} \omega =\int_0^\rho \kappa(x(s)) \dd s$, we observe that three integral terms form the single term
\[
\int_0^\infty  \frac{\dd }{\dd \rho} \left\{  \int_0^\rho  [\nabla \omega(X(s),Y(s))+\omega(X(s)) \omega(Y(s))]\dd s\, e^{ \int_{\varphi(p, [0,\rho])} \omega} \right\}\dd \rho
\]
which vanishes upon integration, thus we arrive at
\begin{align*}
\textrm{Hess} \Lambda(X,Y)=&-[\nabla\omega(X,Y) +\omega(X)\omega(Y) ]\Lambda  -\omega(Y) \p_X \Lambda-\omega(X)\p_Y \Lambda \\
&+\int_0^\infty   \kappa(x(\rho))   \int_0^\rho  \mu(Y(s),X(s)) \dd s\, e^{ \int_{\varphi(p, [0,\rho])} \omega} \dd \rho\\
=&-[\nabla\omega(X,Y) +\omega(X)\omega(Y) ]\Lambda  -\omega(Y) \p_X \Lambda-\omega(X)\p_Y \Lambda \\
&-\int_0^\infty    \mu(Y(\rho),X(\rho)) \, e^{ \int_{\varphi(p, [0,\rho])} \omega} \dd \rho\\
=&-\eta(X,Y) \Lambda  -\omega(Y) \p_X \Lambda-\omega(X)\p_Y \Lambda \\
&-\int_0^\infty   [ \mu(Y(\rho),X(\rho))- \mu(Y(0),X(0))] \, e^{ \int_{\varphi(p, [0,\rho])} \omega} \dd \rho
\end{align*}

Now observe that if $n'=e^f n$, $\Lambda'=e^{-f} \Lambda$, $\omega'=\omega+\dd f$, thus
\begin{align*}
&[\textrm{Hess} \Lambda'+(\nabla\omega' + \omega' \otimes \omega') \Lambda' + \omega' \otimes \dd \Lambda'+\dd \Lambda'\otimes \omega']\\
=& e^{-f}[\textrm{Hess} \Lambda-\textrm{Hess} f+\dd f \otimes \dd f \Lambda-\dd f \otimes \dd \Lambda-\dd \Lambda \otimes \dd f]\\
&+e^{-f}[\nabla \omega+ \textrm{Hess} f+\omega \otimes \omega+\dd f \otimes \omega+\omega\otimes \dd f+\dd f\otimes \dd f]\Lambda\\
&+e^{-f}[\omega \otimes \dd \Lambda - \omega \otimes \dd f \Lambda+\dd f \otimes \dd \Lambda-\dd f \otimes \dd f \Lambda]\\
&+e^{-f}[ \dd \Lambda \otimes \omega - \dd f \otimes \omega \Lambda+ \dd \Lambda\otimes \dd f-\dd f \otimes \dd f \Lambda]\\
=&e^{-f}[\textrm{Hess} \Lambda+(\nabla\omega + \omega \otimes \omega) \Lambda + \omega \otimes \dd \Lambda+\dd \Lambda\otimes \omega]
\end{align*}
As $\mu$ is invariant, we have that the symmetric differential operator
\[
T\Lambda:=\textrm{Hess} \Lambda+\eta \Lambda + \omega \otimes \dd \Lambda+\dd \Lambda\otimes \omega
\]
 changes through multiplication by an exponential factor i.e.\ $T' \Lambda'=e^{-f} T \Lambda$. Of course, this is also the transformation of
\begin{equation} \label{biz}
B(p)(X,Y):=-\int_0^\infty   [ \mu(Y(\rho),X(\rho))- \mu(Y(0),X(0))] \, e^{ \int_{\varphi(p, [0,\rho])} \omega} \dd \rho,
\end{equation}
that is $B'=e^{-f} B$, it is sufficient to take into account that $e^f\dd \rho'=\dd \rho$ as $n'=e^fn$, and that the push forward $X'(\rho)$ of $X(0)$ with $n'$ is going to be of the form $X'=X+gn$ where $g$ satisfies $L_ng=X'(f)=X(f)+g L_nf$ and where we recall that $\mu(\cdot, n)=\mu(n,\cdot)=0$.


We conclude that $\Lambda$ satisfies the equation
\begin{align} \label{kkw}
\begin{split}
&\textrm{Hess}\, \Lambda (X,Y) + \omega(X)  \dd \Lambda(Y)+\dd \Lambda(X) \omega(Y)+\eta(X,Y) \Lambda=B(X,Y),
\end{split}
\end{align}
The same equation shows that $B$ is a $C^0$ symmetric bilinear form, as so is the left-hand side. Its symmetry follows also from the fact that $L_n\mu$ is symmetric.
$\square$ \end{proof}

\section{Bootstrap argument for $L_n\textrm{Ric}\vert_{TH}=0$}

The equation for the Hessian of $\Lambda$ proved in Theorem \ref{bjz} can be used in a bootstrap argument provided the bilinear form $B$ is smooth, which would be the case if it vanishes. From Eq.\ (\ref{bcw}) we see that it vanishes if $\mu$ is left invariant by the flow of $n$. Therefore, the following result is important.

\begin{lemma} \label{drk}
We have the identity
\begin{equation} \label{pfr}
L_n\mu =\frac{1}{2} L_n \textrm{Ric} \vert_{TH}.
\end{equation}
\end{lemma}

\begin{proof}
Let $\{b_a, a=0, \cdots,n\}$ be the basis $(N,n, e_1,e_2, \cdots)$ at $T_p M$, $p\in H$, where $g(n,N)=-1$, $g(n,e_i)=g(N,e_i)=0$, $g(e_i,e_j)=\delta_{ij}$, and let $\{b^a\}$ be the cobasis
\[
(-g(n,\cdot), -g(N,\cdot), g(e_1,\cdot), g(e_2,\cdot), \cdots).
\]
 We have for $X,Y\in T_pH$, using $R(n,X)Y=\mu(X,Y) n$
\begin{align*}
\textrm{Ric}(X,Y)=&\sum_a b^a(R(b_a,X)Y)=-g(n,R(N,X)Y)-g(N,R(n,X)Y)\\
&+\sum_{i=1}^{n-1} g(e_i,R(e_i,X)Y)\\
=&-g(N,R(n,Y)X)-g(N,R(n,X)Y)+\sum_{i=1}^{n-1} g(e_i,R(e_i,X)Y)\\
=&\mu(Y,X) + \mu(X,Y) +\sum_{i=1}^{n-1} g(e_i,R(e_i,X)Y).
\end{align*}
We already know that $L_n\mu$ is symmetric so we need only to show that the Lie derivative of the last term vanishes. Observe that extending $n$ by using the local geodesic flow of $N$ so that $L_nN=0$ we have $L_n(g(N,e_i))=(L_ng)(N,e_i)+g(N,L_ne_i)$, thus
we can extend $e_i$ locally preserving the above properties by imposing $L_n e_i=(L_ng)(N,e_i) n$. Thus let us extend $X,Y$ so that $L_nX=L_nY=0$, we have for each $i$
\begin{align*}
L_n  [g(e_i,R(e_i,X)Y)]=&(L_ng)(e_i,R(e_i,X)Y)+g( e_i, (L_nR)(e_i,X)Y)\\
&+g(L_n e_i, R(e_i,X)Y)+g( e_i, R(L_ne_i,X)Y),
\end{align*}
the last two terms vanish due to $L_ne_i \propto n$ and $R(n,X)Y=\mu(X,Y) n$ (or using the fact that $H$ is totally geodesic). The first term vanishes because $L_n g\vert_{TH\times TH}=0$ and because $R(e_i,X)Y\in TH$ as $H$ is totally geodesic.
We are left with
\begin{equation} \label{oku}
L_n  [g(e_i,R(e_i,X)Y)]=g( e_i, (L_nR)(e_i,X)Y).
\end{equation}
We have \cite[Eq.\ (2.6)]{katzin69} \cite{yano55} denoting $L_ng=h_{ab} e^a \otimes e^b$ and $R(e_d,e_a)e_b=R^c{}_{bda} e_c$,
\begin{equation} \label{goq}
2g_{cm} (L_nR)^m{}_{bda}=(h_{ac;b}+h_{cb;a}-h_{ab,c})_{;d}-(h_{dc;b}+h_{cb;d}-h_{db;c})_{;a}
\end{equation}
We want to show that the right-hand side vanishes for $a,b,c,d\ge 1$. If a covariant $(0,k)$-tensor vanishes on $TH$ then the same holds for $\nabla T \vert_{TH}$, as it is immediate from the formula
\[
(\nabla_X T)(Y,\cdots)=\p_X(T(Y,\cdots))-T(\nabla_X Y, \cdots) - \cdots
\]
using the fact that $\nabla_XY\in TH$ as $H$ is totally geodesic. Since this is true for $L_ng$, the same is true for its two covariant derivatives, and so for the right-hand side of Eq.\ (\ref{goq}). We conclude that the right-hand side of Eq.\ (\ref{oku}) vanishes, and so the desired equation is proved.
$\square$ \end{proof}

\begin{corollary} \label{vja}
Assume $H$ admits a future incomplete generator.
If $L_n\textrm{Ric}\vert_{TH}=0$ on the horizon, then $B=0$. Thus   (cf.\ Eq.\ (\ref{kkv}))
\begin{equation}
\textrm{Hess} \Lambda + \omega\otimes \textrm{d}\Lambda + \textrm{d}\Lambda\otimes\omega + \eta \Lambda =0
\end{equation}
and $\Lambda$ is $C^\infty$ by bootstrapping.
\end{corollary}

We recall that by Lemma \ref{kke} the condition  $L_n\textrm{Ric}\vert_{TH}=0$ does not depend on the choice of $n$ and so represents a property of the horizon.

\begin{theorem} \label{puy}
Assume  the horizon admits a future incomplete generator and that $L_n\textrm{Ric}\vert_{TH}=0$.  Then the smooth vector field $n$ can be chosen so that $\kappa=-1$ and $\Lambda=1$, in which case $L_n\omega=0$, $\eta=0$ (i.e.\ $L_n\nabla\vert_{TH} =0$), $L_n\mu=0$ and for $X,Y\in TH$
\begin{equation}
R(n,X)Y=-[i_Y\nabla_X \omega+\omega(X) \omega(Y)] n
\end{equation}
\end{theorem}


\begin{proof}
We start from any smooth $n$ and rescale it by choosing $f$ in Eqs.\ (\ref{ok4})-(\ref{ok5}) so that $e^f=\Lambda$. Then by Eq.\ (\ref{ok4}), the new vector field, here denoted in the same way, is such that $\Lambda=1$, and by Eq.\ (\ref{biw}) we have $\kappa=-1$. The equation $L_n\omega=0$ follows from Lemma \ref{loq}, while $\eta=0$ follows from Eq.\ (\ref{kkv}) with $B=0$. The equation $L_n\mu=0$ follows from Eq.\ (\ref{pfr}) (it can also be obtained  by taking the Lie derivative of Eq.\ (\ref{kxc}), recalling Eq.\ (\ref{lie}) and using $\eta=0$, $L_n\omega=0$). The equation in display is deduced from the expression of $\mu$ obtained in  Eq.\ (\ref{kxc}).
$\square$ \end{proof}

\section{Smoothness of $\Lambda$ in the non-vacuum case}

In this section we prove the smoothness of $\Lambda$ for the non-vacuum case thus without passing through the bootstrap argument. For another inductive proof see \cite{gurriaran21}.
\begin{definition}
We say that a smooth $s$-dependent covariant tensor field $T_s$
is polynomially bounded in the parameter $s$ if there is a polynomial $p(s)$ such that for every $q\in H$,  $X_i\in T_qH$, $n^*$-horizontal,  $i=1,\cdots, k$,
\[
\vert T_s(X_1,X_2, \cdots, X_k)\vert\le  p(s) \Pi_{i=1}^k \sqrt{g(X_i,X_i)}.
\]
and similarly for some of the $X_i$ on the left-hand side replaced by $n$, in which case the factor $\sqrt{g(X_i,X_i)}$ on the right-hand side has to be omitted.
\end{definition}

\begin{definition}
We say that a smooth $\rho$-dependent covariant tensor field $T_\rho$ belongs to $\mathcal{I}$ if it has the form
\[
T_\rho=  R_\rho\circ \mathcal{P}
\]
where $\mathcal{P}$ is a permutation of the vector arguments, and $R_\rho$ belongs to the smallest family of smooth ($\rho$-dependent) covariant tensors which (a) contains the  smooth ($\rho$-independent) covariant tensor fields, (b) is invariant under tensor products and sums, (c) is invariant under pullback $S_\rho \to \varphi^*_\rho S_\rho$  (d) is invariant under average $S_\rho \to \int_0^\rho S_r \dd r$ (in other words each tensor $R_\rho$ is obtained from a finite family of smooth covariant ($\rho$-independent) tensors by applying a finite number of operations (b), (c) and (d)).
\end{definition}
We notice that if $T_\rho\in \mathcal{I}$ then any  contraction with $n$ belongs to $\mathcal{I}$.

Through the usual splitting argument we get
\begin{lemma}
Every element of $\mathcal{I}$ is polynomially bounded.
\end{lemma}

\begin{proof}
The parameter-independent covariant tensors are polynomially bounded with constant polynomial, as it follows from the compactness of the projective bundle induced by the bundle of $n^*$-horizontal vectors.
We need only to show that (b)-(d) preserve the property of being polynomially bounded. This is trivial for (d), the new polynomial being $\int_0^s p(r)\dd r$. It is also trivial for tensor products or sums, that is for (b), the new polynomial being the product or sum of the original polynomials. As for (c), let $T_s$ be a $(0,k)$-tensor polynomially bounded, and let us study $\varphi_s^* T_s$. We need to split the push forward of the $n^*$-horizontal vector into $n^*$-horizontal and longitudinal parts $\varphi_s{}_* (X_i)=\varphi_s{}_* (X_i)^\perp+\lambda_i(s) n$ as in Eq.\ (\ref{rgf}) where $\lambda_i$ is linearly bounded  $\lambda_i \le K \sqrt{g(X_i,X_i)} s$. Recall also that
\[
g(\varphi_s{}_* (X_i)^\perp,\varphi_s{}_* (X_i)^\perp)=g(\varphi_\rho{}_* (X_i),\varphi_s{}_* (X_i))=g(X_i,X_i).
 \]
 As a consequence, the new polynomial is $(1+Ks)^k p(s)$.
$\square$ \end{proof}

%
%
%
%
%

\begin{lemma}
 Assume $H$ admits a future incomplete generator.  If $\Lambda\in C^k$, $k\ge 1$ and $\nabla^k\Lambda$ can be written in the form
 \[
\nabla^k\Lambda (X_1,X_2, \cdots, X_k)= \int_0^\infty   T_\rho^{(k)}(X_1,X_2, \cdots, X_k) \, e^{ \int_{\varphi(p, [0,\rho])} \omega} \dd \rho
 \]
where $T^{(k)}_\rho \in \mathcal{I}$ is a $(0,k)$-tensor field, then  $\Lambda\in C^{k+1}$ and the previous equation in display holds also for $k\to k+1$ for some $T_\rho^{(k+1)}\in \mathcal{I}$.
\end{lemma}

\begin{proof}

We are going to compute
\begin{equation} \label{qst}
\begin{split}
\nabla^{k+1}\Lambda(X,X_1,X_2, \cdots, X_k)=&\,\p_X [\nabla^k\Lambda (X_1,X_2, \cdots, X_k)]\\ &-\sum_{i} \nabla^k\Lambda (X_1, \cdots, \nabla_X X_i, \cdots, X_k).
\end{split}
\end{equation}
Since $\Lambda \in C^k$ the  latter term on the right-hand side is not problematic. It will be absorbed in terms of the type $\nabla S$ to be introduced below.

The critical term is the former one on the right-hand side. Here we need to show that we can switch differential and integral operators, and this is done by showing that the integral argument obtained by proceeding naively, i.e.\ by operating the switch, is really continuous and integrable, i.e.\ by using the dominated convergence theorem (this theorem is used to show that the limit of the incremental ratio can pass from outside to inside the integral).

We have that (here $L_X$ has to be understood as the Lie derivative with respect to the  vector field   $s\mapsto X(s)$ over the parametrized generator  obtained pushing forward $X=X(0)$ with the flow of $n$, that is, locally it satisfies $L_Xn=0$)
\begin{align*}
\p_X [\nabla^k\Lambda (X_1,X_2, \cdots, X_k)]=&\int_0^\infty   [\p_X ( T_\rho^{(k)}(X_1,X_2, \cdots, X_k))\\
&+ T_\rho^{(k)}(X_1,X_2, \cdots, X_k)\int_0^\rho L_X \omega] \, e^{ \int_{\varphi(p, [0,\rho])} \omega} \dd \rho\\
=&\int_0^\infty   [\p_X  (T_\rho^{(k)}(X_1,X_2, \cdots, X_k))\\
&+ T_\rho^{(k)}(X_1,X_2, \cdots, X_k) [\varphi_\rho^*\omega-\omega](X)] \, e^{ \int_{\varphi(p, [0,\rho])} \omega} \dd \rho
\end{align*}

Now observe that the last term belongs to $\mathcal{I}$ so we can just consider the first term.

Taking into account that any element $T_\rho^{(k)}$ of $\mathcal{I}$ is constructed via some finite number of operations (b)-(d) from parameter independent tensors, we emphasize in the next expressions the contribution of one such tensor $S$. Of course, there will be the contributions of all the smooth  $\rho$-independent  covariant tensors entering the expression for $T_\rho^{(k)}$, but each of them is treated analogously, so they are not displayed in the next expressions.

In the expression for $\p_X  (T_\rho^{(k)}(X_1,X_2, \cdots, X_k))$ the derivative is going to distribute over each of the parameter-independent tensors. When we evaluate $T_\rho^{(k)}$ on the vectors $(X_1, \cdots, X_k)$, some vectors $X_p,\dots$, apply to the entries of $S$, but due to the pullbacks $\varphi^*_{r}$ that enter the expression of  $T_\rho^{(k)}$ before $S$, these vectors are really pushed forward by some cumulative parameter $s$ before being evaluated on $S$. In other words $S$ will contribute to $\p_X  (T_\rho^{(k)}(X_1,X_2, \cdots, X_k))$ with a term of the form
\[
\cdots\p_{\varphi_s{}_* (X)} [S(\varphi_s{}_* (X_p), \cdots)]\cdots.
\]
This  derivative has to be converted into a covariant derivative by using the last terms in Eq.\ (\ref{qst}). Those involving $S$ have the form
\[
\cdots-S(\varphi_s{}_* (\nabla_{X}X_p), \cdots)\cdots.
\]
We are going to face here the usual difficulty that the push forward of covariant derivatives is not the covariant derivative of the push forwards.
Taking into account the difference provided by Eq.\ (\ref{vqq}), the sum of the previous two expressions in display gives
\begin{align*}
&\cdots\nabla S(\varphi_s{}_* (X),\varphi_s{}_* (X_p), \cdots)+S( \nabla_{\varphi_s{}_*(X)}\varphi_s{}_*(X_p) -\varphi_s{}_*(\nabla_XX_p),\cdots)+\cdots\\
&=\cdots\nabla S(\varphi_s{}_* (X),\varphi_s{}_* (X_p), \cdots)+S( n,\cdots)\int_0^s \dd t \varphi_t^* \eta(X,X_p)+\cdots
\end{align*}

These terms are of type $\mathcal{I}$ and thus polynomially bounded.

Therefore the claim is proved and the dominant convergence theorem can indeed be applied as any polynomial multiplied by the exponential $e^{ \int_{\varphi(p, [0,\rho])} \omega} =e^{ \int_0^\rho \kappa(x(s))\dd s}$ is integrable (remember that there is some $\epsilon>0$ such that $\kappa<-\epsilon$ on $H$).
$\square$ \end{proof}

\begin{corollary}
Assume $H$ admits a future incomplete generator. Function $\Lambda$ is smooth.
\end{corollary}

\begin{proof}
The induction step is proved by the previous lemma. The starting point of the induction is justified by the expression of $\Lambda$ (or by that of its differential).
$\square$ \end{proof}

\section{General properties in the non-vacuum case} \label{nnn}

In this work by ``non-vacuum'' we mean ``not necessarily Ricci flat''. Of course, the Ricci flat condition is included in our study.
We recall that $H$ and the spacetime satisfy conditions ($\star$) and ($\star\star$) introduced previously.

We also recall that in the next result $\eta$ controls the Lie derivative of the connection, Eq.\ (\ref{lie}).

\begin{theorem} \label{ppq}
Assume  $H$  admits a future incomplete generator.  The smooth vector field $n$ can be chosen such that $\kappa=-1$ and $\Lambda=1$, in which case $L_n\omega=0$. Moreover,
\begin{align} \label{cqy}
\eta=B=&-\frac{1}{2}\int_0^\infty (\varphi^*_\rho \textrm{Ric}-\textrm{Ric})\, e^{-\rho} \dd \rho=-\frac{1}{2}\int_0^\infty \varphi^*_\rho (L_n\textrm{Ric})\, e^{-\rho} \dd \rho .
\end{align}
\end{theorem}

\begin{proof}
It is sufficient to start from any smooth $n$ and choose $f$ in Eqs.\ (\ref{ok4})-(\ref{ok5}) so that $e^f=\Lambda$, while using Eq.\ (\ref{biw}). The equation $L_n\omega=0$ follows from Lemma \ref{loq}. The last equation follows from Lemma \ref{drk}, from Eqs.\ (\ref{kkv})-(\ref{bcw}) through integration by parts.
$\square$ \end{proof}

We stress that in the following results we are not assuming $\kappa$ constant unless otherwise specified.

\begin{definition} \label{der}
The following properties are equivalent. If they hold we say that $H$ is {\em future non-degenerate}, and similarly in the past version. If they do not hold in any time orientation, we say that $H$ is {\em degenerate}.
\begin{enumerate}
\item $H$ admits a future incomplete generator (and hence every generator is future incomplete),
\item The smooth vector field $n$ can be chosen such that $\kappa<0$ over $H$,
\item The smooth vector field $n$ can be chosen such that $\kappa=-1$ over $H$,
\item For some choice of smooth vector field $n$ (and hence for  every choice) there is a generator over which $\int_0^\infty \kappa(s)\dd s=-\infty$ (and hence the same is true for every generator).
\end{enumerate}
\end{definition}

\begin{proof}[of the equivalence]
Observe that 1 $\Rightarrow$ ``2 and 3'' follows from Theorem   \ref{ppq}, while ``2 or 3'' $\Rightarrow$ 1 follows from the continuity of $\kappa$ and compactness of $H$, which implies $\kappa<-\epsilon<0$ over $H$, so that  the integral (\ref{bud}) converges. We have shown the equivalence of 1, 2 and 3.

Property 4 does not depend on the choice of smooth vector field $n$ since under a gauge change $n'=e^f n$, $\kappa'=e^f(\kappa + \p_n f)$ and $\dd s'=e^{-f} \dd s$, thus over every segment $\eta$ of a generator with endpoints $p,q\in H$, $\int_\eta \kappa'\dd s'=\int_\eta \kappa \dd s+\Delta f$, where $\Delta f=f(q)-f(p)$ is bounded, as $f$ is bounded.

Assume 4 then by Prop.\
\ref{vos} the same property holds for every generator, and by Thm.\ \ref{bhq} property 2 holds. For the converse, clearly under 3 property 4 holds.
$\square$ \end{proof}

\begin{proposition}
For a future non-degenerate $H$ all generators are complete in the past direction (and incomplete in the future direction). For a degenerate $H$ all generators are complete (in both directions).
\end{proposition}

\begin{proof}
In Corollary \ref{cbv} we already proved the statement in parenthesis.
Let $n$ be such that $\kappa=-1$, then for $-n$ we have $\kappa=1$. Observe that the integral curves for $-n$ are past-directed, thus the claim follows from Eq.\ (\ref{bud}). The last statement follows from the fact that 1 in Def.\ \ref{der} does not hold  neither in the future nor in the past cases.
$\square$ \end{proof}

The (embedded) $C^2$ null hypersurfaces  are known to be locally achronal \cite{galloway00} \cite[Thm.\ 6.7]{minguzzi18b}, thus for them we can infer the existence of a neighborhood in which they are achronal accordingly to the next more general result
\begin{proposition} \label{ptg}
Let $N$ be a locally achronal topological co-dimension one submanifold which is closed in the topology of $M$, then there is an open neighborhood $U\supset N$ such that $N$ is achronal in $U$.
\end{proposition}


\begin{proof}
Let $V$ be a global smooth timelike vector field, and let $x\in N$. By local achronality and the other assumptions on $N$, $x$ admits a neighborhood $W$ in which $N\cap W$ is an achronal boundary for the spacetime $W$.
There is a cylindrical coordinated neighborhood $C\subset W$, such that  $N\cap C$ coincides with the image of a locally Lipschitz graph $x^0=h({\bm x})$, with  $\p_0=V$ (cf.\ \cite{penrose72}\cite[Thm. \ 2.87]{minguzzi18b}).
 Hence, due to the time-orientation of the spacetime, $N$ is two-sided.
 This means that any sufficiently small tubular neighborhood $U$ can be written as the union of two `one-sided' neighborhoods $U_1,U_2$, $U_1\cap U_2=N$, where a global future-directed timelike vector field would point from $N$ towards $U_1$ (i.e.\ $U_1$ is on the local future side of $N$, while $U_2$ is on the local  past side). Suppose that there is a $C^1$ timelike curve $\gamma\subset U$ connecting two points of $N$. By shortening it if necessary, we can assume that it intersects $N$ just at the endpoints $p$ and $q$ (by local achronality $p\ne q$). As a consequence, $\gamma\backslash\{p,q\}$ is contained in either $U_1\backslash N$, or $U_2\backslash N$, as it is not possible to pass from one side of the neighborhood to the other without crossing $N$. Without loss of generality, let us assume the former possibility.
 Picking $r\in \gamma\backslash\{q\}$ in a convex neighborhood of $q$, we can replace the last piece of $\gamma$ connecting $r$ to $q$ with a timelike geodesic segment. Not all the timelike geodesic segment can be in $U_2$, otherwise $r\in N$, contradicting the definition of $q$. Thus, without loss of generality, we can find a piecewise $C^1$ timelike curve $\sigma$ in $U_1$, intersecting $N$ just at the endpoints $p,q'$, and having a last geodesic segment connecting $r \in U_1\backslash N$ to $q'$, $r\ne q'$.


 By the local achronality of $N$, $q'$ admits a convex neighborhood such that the exponential map of the past timelike cone at $q'$ on that neighborhood, does not intersect $N$, and hence is entirely contained in $U_2\backslash N$. This implies that $\sigma$ intersects $U_2\backslash N$, a contradiction.

 The contradiction proves that $N$ is achronal in $U$.
$\square$ \end{proof}

The next result shows that the surface gravity could have been introduced in a different way, which is indeed the original one in \cite{moncrief83}.
\begin{proposition} \label{bkq}
Let $T$ be a vector field defined in a neighborhood of $H$ such that  $g(n,T)=-\frac{1}{a}=const\ne 0$ on $H$ (hence transverse to $H$), and extend $n$ to a neighborhood of $H$ in such a way that $L_Tn=0$, then
\begin{equation} \label{yry}
\kappa=\frac{a}{2} \p_T g(n,n) \vert_H.
\end{equation}
\end{proposition}

\begin{proof}
We have
\[
\nabla_T g(n,n)=2 g(\nabla_T n, n)=2 g(\nabla_n T,n)=2[\p_n g(T,n)-g(T,\nabla_n n)]=\frac{2\kappa}{a}  .
\]
$\square$ \end{proof}

Typically $T$ will be  future-directed causal, hence $a>0$.

In the previous references results similar to the following ones, but under some achronality assumptions, were given a proof based on Gaussian null coordinates \cite{petersen18b,petersen19}. Our proof is topological and direct.
\begin{theorem} \label{vcw}
Let $T$ be a future-directed causal vector field transverse to $H$ such that $g(n,T)=cost$, and let $\psi_t$ be its flow.
If $\kappa<0$ all over $H$, then for sufficiently small $\vert t \vert$, $H_t:=\psi_t(H)$ is a compact timelike hypersurface for  $t>0$, and a compact spacelike hypersurface for $t<0$. Moreover, no two hypersurfaces in the family $\{H_t\}$ intersect.
\end{theorem}

\begin{proof}
Let $p\in H$, we are going to study the causal character of $T_{\psi_t(p)} \psi_t(H)=\dd \psi_t (T_p H)$ for sufficiently small $\vert t\vert$, proving the claim at $p$. All the geometric objects constructed in what follows can be chosen so as to be continuous with $p$, then the result follows from  the compactness of $H$ by a standard argument.

Since $\kappa(p)<0$ (we do not need constancy of $\kappa$) we have from Eq.\ (\ref{yry}), recalling that $L_Tn=0$ and hence that $n$ is invariant under the flow, $g(\dd \psi_t(n),\dd \psi_t(n))<0$ for sufficiently small $t>0$. As the induced metric on $T_{\psi_t(p)} \psi_t(H)$ can only be of three types we conclude that the Lorentzian case applies.

For the case $t<0$, observe that $\p_Tg(n,n)=(L_T g)(n,n)$, thus by continuity there is an open double cone  $D_p\subset T_pH$ containing $\{n,-n\}$ such that the quadratic form $(L_Tg)$ is negative on $D_p$.
More precisely, introduced a Riemannian metric $h$ on $H$,  there is $\epsilon>0$ such that for every $v\in D_p$, $(L_T g)(v,v)<-\epsilon h(v,v)$.

Let $R_p=T_pH\backslash [D_p\cup \{0\}]$, then $R_p$ does not intersect $C_p\subset T_p M$, the causal cone at $p$. As a consequence for sufficiently small $\vert t\vert$, $\dd \psi_t (R_p)\cap C_{\psi_t(p)}=\emptyset$ which means that $\dd \psi_t (R_p)$ consists of spacelike vectors.

Since the zero vector is push forwarded to the zero vector, we need only to study $\dd \psi_t (R_p)$, as $T_pH=R_p\cup D_p\cup \{0\}$.
But regarding $g$ as a quadratic form,  $g\vert_{D_p}\ge 0$ and for $v\in D_p$
\[
\frac{\dd }{\dd t} g(\dd \psi_t(v),\dd \psi_t(v))\vert_{t=0}=(L_Tg)(v,v)< 0
\]
thus $g(\dd \psi_t(v),\dd \psi_t(v))=g(v,v)+(L_Tg)(v,v) t +o_v(t)$. Note
that $L_TL_Tg\vert_{TH}$ is bounded over the compact sphere of unit vectors  of $h$, so $\vert o_v(t)\vert \le b t^2 h(v,v)$ for some $b>0$, i.e.\ it can be uniformly controlled over $D_p$. This implies  that for sufficiently small $\vert t\vert$ with $t<0$, $\dd \psi_t(D_p)$ consists of spacelike vectors.

The last statement follows from Prop.\ \ref{ptg} as the existence of $U$ implies that for $p,q\in H$,  it cannot be $\psi_{t'}(p)=\psi_{t}(q)$,  $t'\ne t$,  as we would get $\psi_{t'-t}(p)=q$ which leads to a contradiction with  the achronality of $H$.
$\square$ \end{proof}

We recall that a {\em temporal function} $f$ is a $C^1$ function whose gradient is timelike and past-directed. Equivalently, $\ker \dd f$ is spacelike and $\dd f(T)>0$ with $T$ future-directed causal vector field. Temporal functions are {\em time functions}, that is, they increase over every causal curve.

In the next proposition it is assumed that $\vert t\vert \le c$ where $c$ is so small that the previous result applies for every $t$ in this interval.
\begin{proposition} \label{gol}
Let $\kappa<0$ all over $H$ and let $T$ be as in the previous result.
The function $t$ defined in a neighborhood of $H$ by $\p_t=T$, $t\vert_H=0$, is a temporal function for $t< 0$. Moreover, every inextendible causal curve that intersects one $H_{t'}$, $t'<0$ must intersect all the other $H_t$, $t<0$.
\end{proposition}

\begin{proof}
Observe that $\dd t(T)=1$ and $\ker \dd t(\psi_\tau(p))=T_{\psi_\tau(p)} H_{\tau}$ which is spacelike by Thm.\ \ref{vcw}, thus $t$ is temporal on the region $t<0$. The last statement follows from the fact that no inextendible causal curve can accumulate on a compact spacelike manifold, for the limit curve theorem would give a causal curve in that manifold, a contradiction.
$\square$ \end{proof}

For a non-degenerate $H$, chosen $n$ such that $\kappa=-1$, as $L_n\omega=0$ the following (Petersen) Riemannian metric\footnote{A  study of Riemannian metrics of this form over null hypersufaces can be found in \cite{gutierrez16}.} \cite{petersen18b,bustamante21} is
left invariant by the flow $\varphi$
\begin{equation}
\sigma:=\tilde g + \omega \otimes \omega
\end{equation}
that is, $L_n \sigma=0$.

We recall that  given a generator $\gamma$, the future omega-limit set is defined by \cite{minguzzi07f}
\[
\Omega_f(\gamma)=\bigcap_t \overline{\gamma([t,+\infty))}
\]
that is, it is the set of future accumulation points of $\gamma$. A similar definition is given for  $\Omega_p(\gamma)$.
\begin{proposition} \label{vqa}
Let $H$ be non-degenerate, then for every generator $\gamma$ we have $\overline{\gamma}= \Omega_f(\gamma) \cap \Omega_p(\gamma)$.
\end{proposition}

In other words $\gamma$ future (and past) accumulates to each of its points (compare with the similar result in \cite{moncrief20}).
\begin{proof}
Let us choose $n$ so that $\kappa=-1$.
Let $p\in \gamma$ and let $B(p, r)$ be the open ball of $\sigma$-radius $r$ centered at $p$.

Since the $\sigma$-volume of $H$ is finite and since the volume is preserved by the flow $\varphi$, by the standard application of the Poincar\'e recurrence argument, for each $\tau>0$ there is some $k\in \mathbb{N}\backslash \{0\}$ such that $\varphi_{k\tau}(B(p,r))\cap B(p,r)\ne \emptyset$. As $\tau$ is arbitrary this means that the future-directed generator starting from $p$
intersects $B(p,2r)$ indefinitely  in the future. As $r$ is arbitrary, $p$ is a future accumulation point of  $\gamma$, and similarly in the past case.

This shows $\gamma \subset \Omega_f(\gamma) \cap \Omega_p(\gamma)$ and since the right-hand side is closed   $\overline{\gamma} \subset \Omega_f(\gamma) \cap \Omega_p(\gamma)$. The other direction is clear since, by definition, $\Omega_f(\gamma) \subset \overline{\gamma}$ and similarly in the past case.
$\square$ \end{proof}

The next result was proved in \cite{moncrief83} for vacuum spacetimes under an analyticity  assumption and for closed generators.
See also \cite[Cor.\ 2.13]{petersen19} for the smooth vacuum case of point (ii).

We recall that the region of chronology violation $\mathcal{C}$ consists of those points $p$ through which passes a closed timelike curve.
The set $[p]_U$ is the chronological class of $p$ for the spacetime $U$ in the induced metric.

\begin{theorem} \label{qkw}
Let $H$ be future non-degenerate, let $T$ be a future-directed causal vector field transverse to $H$ such that $g(n,T)=cnst<0$,
and let $H_t:= \psi_t(H)$.
\begin{itemize}
\item[(i)] For each sufficiently small $t>0$, $H_t$ is contained in the chronology violating set $\mathcal{C}$, hence $H\subset \overline{\mathcal{C}}$. More precisely, there is a neighborhood $U\supset H$, such that for each $p\in I^+(H,U)$, $[p]_U=I^+(H,U)$.\\
\item[(ii)] For each neighborhood $U$ of $H$ in which $H$ is achronal (it exists by Prop.\ \ref{ptg}) we have, for each sufficiently small $t<0$, that the compact spacelike hypersurface $H_t\subset U$ is acausal in $U$ and such that $H^+(H_t,U)=H$ and $D^+(H_t, U)=\psi(H,[t,0))=\cup_{s\in [t,0)} H_s$.
\end{itemize}
Thus every non-degenerate horizon $H$ is actually a  Cauchy horizon bounded on one side by a (one single chronology violating class of a) region of chronology violation.
\end{theorem}

Remember that a partial Cauchy hypersurface is an acausal edgeless  set \cite{minguzzi18b}. Here $H_t$ is just a `local' partial Cauchy hypersurface  as it is acausal just in $U$. However, `local' can be dropped if $H$ is achronal. Notice also that $H_t$ is diffeomorphic to $H$ as it is standard between a (local) partial  Cauchy hypersurface and its Cauchy horizon.

The versions of (i) and (ii) for past non-degenerate horizon can be obtained reversing the time orientation. The region of chronology violation will be found in the past of $H$, while the partial Cauchy hypersurfaces with past Cauchy horizon $H$ will be found in the future of $H$.

\begin{proof}
$(i)$. Every point of $\psi_t(H)$ is of the form $p':=\psi_t(p)$ for some $p\in H$. Let $r<p$ be a point slightly before $p$ in the generator $\gamma$ passing through $p$. The curve $\gamma$ accumulates to the future to each point of $\gamma$, hence also to $r$. Since $\psi_t\vert_H$ is a diffeomorphism for sufficiently small $\vert t \vert$, the curve $\gamma':=\psi_t \circ \gamma$ accumulates to the future on $r':=\psi_t(r)$.
But for  $t>0$ since $\kappa<0$, we have that  $\dd \psi_t(n)$ is timelike, thus the curve $\gamma'$ is timelike, which implies $r'\ll p'$. Moreover, since $\gamma'$ passes through $p$ and accumulates to $r'$, we have, by the openness of the chronological relation that there is a closed timelike curve passing through $p'$ which gives the desired result.

The chronological classes are open and their union gives the chronology violating set. Since $H$ is connected, $\psi(H, (0, a))$ is connected and contained in the chronology violating class of $U=\psi(H, (-a , a))$ for sufficiently small $a>0$, hence it is entirely contained in one chronology violating class of $U$.

$(ii)$. By Prop.\ \ref{gol} we know that every inextendible causal curve $\gamma$ passing through $q\in H_{t'}$, intersects $H_{t}$, $t<t'<0$. This fact and point (i) show that  $D^+(H_t, U)=\psi(H,[t,0))=\cup_{s\in [t,0)} H_s$ and $H^+(H_t,U)=H$.
Finally, observe that if there are
$p_1, p_2\in H_t$ connected by a future-directed causal curve $\gamma\subset U$  then $H$ is  non-achronal in $U$, due to the fact that  $\gamma$ intersects $H$ and the timelike integral curve of $T$ passing through $p_2$ intersects $H$.
$\square$ \end{proof}

\subsection{Classification in the 3+1 dimensional case}

By using the invariance of the Petersen metric $\sigma$ under the flow $\varphi$, Bustamante and Reiris  obtained a classification for the topology and for the orbital types
of the null generators  of the compact non-degenerate Cauchy horizons in smooth vacuum 3 + 1-spacetimes \cite{bustamante21}. This classification improved a previous classification in \cite{moncrief20}. As they also stress in their paper, the proof relies only on the fact that the surface gravity can be normalized to $-1$. Indeed, the proof uses  results on isometric actions for  Riemannian 3-dimensional manifolds, there applied to the case of $(H,\sigma)$.

According to our previous results, their theorems generalize to the non-vacuum case as follows (we introduce all the assumption on our horizon $H$ for  clarity)

\begin{theorem}
Let $(M,g)$ be a 3+1-dimensional spacetime which satisfies  the dominant energy condition.
Let $H$ be a smooth compact totally geodesic  non-degenerate horizon (hence generated by lightlike geodesics).
Then
the classification of Theorem 1.1 and Corollary 1.2 of \cite{bustamante21} applies to $H$.
\end{theorem}

\section{Conclusions}
In this work we explored properties related to surface gravity of ($\star$) a compact connected smooth totally geodesic null hypersurface $H$ under ($\star\star$) the dominant energy condition (or some weaker condition).
These hypersurfaces arise naturally as compact Cauchy horizons on spacetimes satisfying ($\star\star$).
We stressed that the dominant energy condition, rather than  the vacuum assumption, is sufficient in this connection for many purposes. For instance, the results (a)  `a horizon $H$ that admits an incomplete generator admits a lightlike tangent field $n$ whose surface gravity is a constant' (non-degenerate horizon), and (b) `every non-degenerate horizon is a (local) Cauchy horizon bounded by a region of chronology violation', can be obtained without a  vacuum assumption. Much of the work was in fact devoted to proving (a) with a strategy largely independent of that of \cite{moncrief08,reiris21}. Our method of proof gives new  insights into this type of normalization problems.

\section*{Appendix: Existence and regularity of the homotopy} \label{hom}

In this appendix we prove that ribbons with arbitrary starting horizontal curve $\sigma_0$ exist, that they are smooth and that they can be arbitrarily extended in the longitudinal direction. Furthermore, these ribbons are foliated by horizontal curves that establish a bijection between the longitudinal sides $\gamma_0$ and $\gamma_1$.

Let $\sigma_0:[0,1] \to H$ be a horizontal curve, and let us consider the map
\[
\chi: [0,1]\times \mathbb{R} \to H, \qquad \chi(r,u)=\varphi(\sigma_0(r), u),
\]
where $\varphi$ is the flow of $n$, so that $\dd \varphi(r, u)(\p_u)=n$.

The 1-form on $[0,1]\times \mathbb{R}$ given by $\beta= \varphi^*(n^*)$, satisfies $\beta(\p_u)=n^*(\dd \varphi(\p_u))=n^*(n)=1$ and hence has the form
\[
\beta=\dd u-U(r,u) \dd r
\]
where $U$ is some smooth function. Since $\sigma_0$ is horizontal, we have $U(r,0)=0$. The distribution $\ker \beta$ is integrable on  $[0,1]\times \mathbb{R}$, and the leaves are graphs of maps $r\mapsto (r,s(r))$ with
\[
\frac{\dd s}{\dd r}=U(r,s).
\]
For each starting point $(0,s(0))$ we have one integral leaf. Since the solutions to the above ODE exist and are unique, distinct leaves do not intersect. The image of the leaves under the map $\chi$ are the horizontal curves that foliate the ribbon.

We want to prove that every leaf starting from  $\{0\} \times \mathbb{R}$ reaches $\{1\}\times \mathbb{R}$, in other words, each leaf, regarded as a graph, has an $r$-domain which coincides with $[0,1]$. Since the solutions to the ODE are unique the integral leafs would end up establishing a bijection $s(0)\leftrightarrow s(1)$. Notice that each leaf would be $\varphi$-mapped to a horizontal curve connecting $\gamma_0$ to $\gamma_1$.


Observe that the map $\mu: r\mapsto \chi(r,s(r))=\varphi(\sigma_0(r),s(r))$ has as image a horizontal curve, which means $n^* (\dd \mu (\p_ r))=0$. Since
\[
\dd \mu (\p_r)=\left(\dd \varphi_{s(r)}\right) (\sigma_0'(r))+ s'(r) n ,
\]
we arrive at the ODE
\begin{equation} \label{ode}
s'(r)=-n^*\left(\dd \left(\varphi_{s(r)}\right) (\sigma_0'(r))\right),
\end{equation}
which gives a precise form to the function $U$ introduced above.

\begin{lemma}
Let $\sigma_0:[0,1]\to H$ be a horizontal curve. For every $\rho \in \mathbb{R}$, the maximal solution to the ODE (\ref{ode}) with initial condition $s(0)=\rho$ is defined on $[0,1]$.
\end{lemma}
Observe that since $\sigma_0$ is horizontal the zero map $s(r)=0$ is a solution, and since no two solutions intersect, we have that $\rho>0$ implies $s(r)>0$ for every $r$ in the domain of the solution, and similarly for $\rho<0$.

\begin{proof}
Let us consider the case $\rho>0$, the other case being analogous.
Let $s: I \to \mathbb{R}$ be the maximal solution of the ODE, with $I=[0,\alpha)$, $\alpha \le 1$, the maximal interval.
By Lemma \ref{nwx} there is $K>0$ such that
\[
s'(r)\le \vert s'(r)\vert= \big\vert n^*\left(\dd \left(\varphi_{s(r)}\right) (\sigma_0'(r))\right) \big\vert \le K \sqrt{g(\sigma_0'(r),\sigma_0'(r))} s(r)\le C s(r)
\]
where $C>0$ is a suitable constant.
Then  the solution is bounded\footnote{To see this, define $h(r)=s(r) e^{-Cr}$. We have $h'=(s'-Cs) e^{-Cr}\le 0$, thus $h$ is decreasing, hence $h(r)\le h(0)=s(0)$. Thus $s(r)\le s(0) e^{Cr}\le \rho e^{Cr}$.} on
$I$ and hence, by ODE theory \cite[Cor.\ 3.1]{hartman02}, $I=[0,1]$.
\end{proof}




\subsection{Local horizontal lift}
In this section we give some details on the notion of horizontal lift.

Let $p\in H$ and, on the manifold $H$, let $\{x^a, a=0,\cdots, n-1\}$ be the coordinates of a cylinder coordinate neighborhood $C$ of $p$ such that $\p_0=n$ and $x^0=0$ is a spacelike codimension-2 manifold.
Remember that $n^*$ is a 1-form on $H$ such that $n^*(n)=1$, thus its kernel is a subspace of $TH$ transverse to $n$.
This means that in local coordinates
\[
n^*=\dd x^0+A_i(x^0, x^i) \dd x^i
\]
where $i=1,\cdots, n-1$. This is the typical connection of generalized gauge theories \cite{mangiarotti84,modugno91}
 (in standard gauge theories $A_i$ would not depend on $x^0$, that is $n^*$ would be invariant under the flow $\varphi$ of $n$).


The condition of $n^*$-horizontality for a curve $r\mapsto x(r)$ on $H$ is
\[
\frac{\dd x^0}{\dd r}=-A_i(x^0, x^i) \frac{\dd x^i}{\dd r}
\]
which is a first order ODE for $x^0(r)$ once the map $r\mapsto x^i(r)$ has been assigned. The curve  $\alpha: r\mapsto x^i(r)$ lives on the quotient space $Q:=C/n$ of the cylinder by the flow, and what we are defining would be called {\em horizontal lift} in generalized gauge theories.

For each  $q_1\in Q$, and curve  $\alpha$ with starting point $q_1$ in $Q$,
we have, for each point in the fiber of $q_1$, one and only one horizontal lift. This existence and uniqueness
follows from the existence of solutions to the previous ODE \cite{hartman02}. The smooth dependence of the horizontal lift on the choice of point in the fiber follows from regularity results on the dependence of ODE from the initial conditions \cite{hartman02}.  Since solutions to the ODE are unique, different horizontal lifts of the same curve do not intersect.

\section*{Acknowledgments}
S.G.\ thanks the Department of Mathematics of Firenze for kind hospitality.
E.M. was partially supported by GNFM of INDAM.




%

\end{document}